\documentclass[aps,pra,reprint,superscriptaddress]{revtex4-1}

\renewcommand\thesection{\arabic{section}}
\renewcommand\thesubsection{\arabic{section}.\arabic{subsection}}
\makeatletter
\renewcommand\p@subsection{}
\makeatother

\makeatletter
\renewcommand\p@subsubsection{}
\makeatother

\usepackage[a4paper,left=2cm,right=2cm,top=3cm,bottom=3cm]{geometry}
\usepackage{soul}
\usepackage{relsize}
\usepackage{lmodern}
\usepackage{slantsc}
\usepackage{graphicx}
\usepackage{amsmath}
\usepackage{amssymb}
\usepackage{amsthm}
\usepackage{amsbsy}
\usepackage{mathrsfs}
\usepackage{varioref}
\usepackage{dsfont}
\usepackage{bm}
\usepackage{color}
\usepackage[colorlinks = false, citecolor=green, linkcolor=red, breaklinks=true, linktocpage=true]{hyperref}
\usepackage[font = footnotesize, labelfont = bf, justification = RaggedRight]{caption}
\usepackage{subfigure}
\usepackage{booktabs} 
\usepackage{array} 
\usepackage{paralist} 
\usepackage{verbatim} 
\usepackage[parfill]{parskip} 
\usepackage{natbib}

\newcommand{\sub}[1]{_{\!\mathsmaller{#1}}}
\newcommand{\subw}[1]{_{\!\mathsmaller{\, #1}}}
\newcommand{\eq}[1]{Eq.~\eqref{#1}}
\newcommand{\fig}[1]{Fig.~\ref{#1}}

\newcommand{\app}[1]{Appendix~(\ref{#1})}
\newcommand{\lemref}[1]{Lemma~\ref{#1}}

\newcommand{\thmref}[1]{Theorem~\ref{#1}}

\newcommand{\defref}[1]{Definition~\ref{#1}}

\newcommand{\featref}[1]{Feature~\ref{#1}}
\newcommand{\reqref}[1]{Requirement~\ref{#1}}

\newcommand{\<}{\langle}
\renewcommand{\>}{\rangle}
\newcommand{\ket}[1]{|{#1}\rangle}

\newcommand{\pr}[1]{P[{#1}]}
\newcommand{\prs}[1]{P\sub{\mathcal{S}}[{#1}]}
\newcommand{\prdd}[1]{P\sub{\mathcal{D}}[{#1}]}
\newcommand{\prw}[1]{P\sub{\mathcal{W}}[{#1}]}

\newcommand{\h}{{\mathcal{H}}}

\newcommand{\rr}{{\mathcal{R}}}

\newcommand{\s}{{\mathcal{S}}}
\newcommand{\ee}{{\mathcal{E}}}

\newcommand{\dd}{{\mathcal{D}}}
\newcommand{\ii}{{\mathcal{I}}}
\newcommand{\ww}{{\mathcal{W}}}

\newcommand{\xx}{{\mathcal{X}}}
\newcommand{\co}{\mathds{C}}

\renewcommand{\nat}{\mathds{N}}

\newcommand{\one}{\mathds{1}}
\newcommand{\zero}{\mathds{O}}

\newcommand{\tr}{\mathrm{tr}}

\theoremstyle{plain}
\newtheorem{defn}{Definition}

\newtheorem{lem}{Lemma}
\newtheorem{thm}{Theorem}
\newtheorem{cor}{Corollary}

\newtheorem{feat}{Feature}
\newtheorem{req}{Requirement}

\numberwithin{equation}{section}

\begin{document}

\title{A quantum Szilard engine without heat from a thermal reservoir}
\author{M. Hamed Mohammady}
\affiliation{Department of Physics and Astronomy, University of Exeter, Stocker Road, Exeter, EX4 4QL, United Kingdom}
\author{Janet Anders}
\affiliation{Department of Physics and Astronomy, University of Exeter, Stocker Road, Exeter, EX4 4QL, United Kingdom}


\begin{abstract}
We study a quantum Szilard engine that is not powered by heat drawn from a thermal reservoir, but rather by projective measurements. The engine is constituted of a system $\s$, a weight $\ww$, and a Maxwell demon $\dd$, and  extracts work via  measurement-assisted feedback control.  By imposing  natural constraints on the measurement and feedback processes, such as energy conservation and leaving the memory of the demon intact, we show that while the engine can function without heat from a thermal reservoir, it must give up at least one of the following features that are satisfied by a standard Szilard engine: (i) repeatability of measurements; (ii) invariant weight entropy; or (iii) positive work extraction for all measurement outcomes. This result is shown to be a consequence of the Wigner-Araki-Yanase (WAY) theorem, which imposes restrictions on the observables that can be measured under additive conservation laws. This observation is a  first-step towards developing ``second-law-like'' relations for measurement-assisted feedback control beyond thermality.  
\end{abstract}

\maketitle

\section{Introduction}

The possibility of extracting work from a system that is in thermal equilibrium, by means of measurement-assisted feedback control \cite{Sagawa-Feedback-Control-Classical,Sagawa-Measurement-Feedback-Autonomous}, was first introduced by Maxwell \cite{Maxwellian-demon,Maxwells-demon-colloquium}. Seemingly violating the second law of thermodynamics, this observation sparked an intense debate, with a key contribution coming from Leo Szilard \cite{Szilard}. Szilard envisioned an engine where the system, $\s$, is a single particle in a box of volume $V$. Maxwell's demon, $\dd$, extracts work from the system by performing two operations, namely, measurement and feedback. During the measurement stage, the demon places a frictionless partition inside the box, thus dividing it into two volumes $V\sub{L}$ and $V\sub{R}$. Thereafter, the demon measures on which side the particle is located. During the feedback stage, conditional on the particle being found on the right (left) side of the partition, the demon attaches a weight-and-pulley mechanism to the right (left) of the partition so that, as the particle collides with the partition, the weight is elevated. The increase in the weight's gravitational potential energy is identified as the extracted work. This is shown schematically in \fig{Szilard engine}. 

\begin{figure}[!htb]
\includegraphics[width= 8 cm]{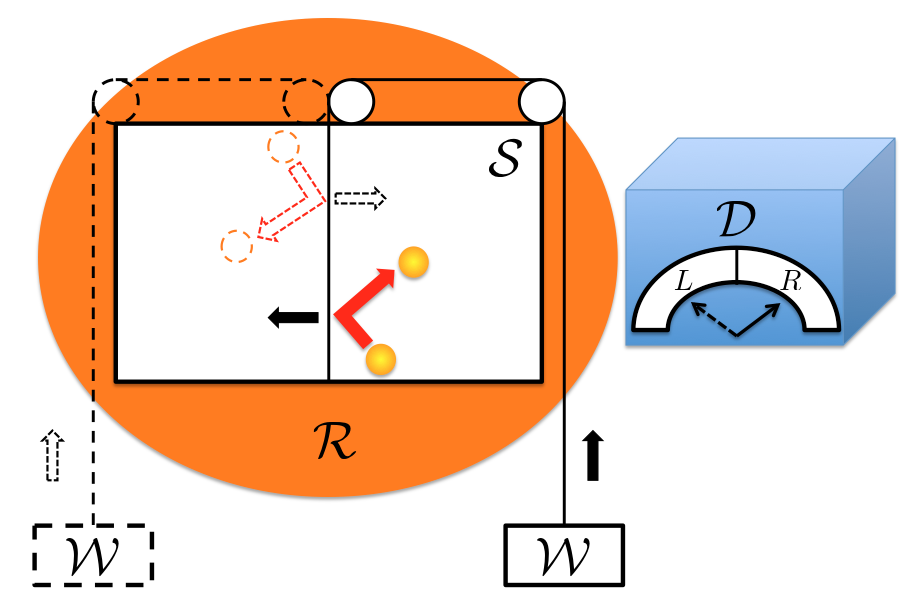}
\caption{ Szilard's engine. The demon, $\dd$, places a partition inside a box containing a single particle. This is the system $\s$. During the measurement stage the demon  measures the system and determines that the particle is on the right (left) hand side. This is stored in the demon's memory as the state $R$ ($L$). During the feedback stage, the demon  attaches a weight $\ww$ to the partition via a pulley mechanism placed on the right (left) hand side. As the particle collides with the partition, moving it to the left (right), the weight is elevated and thus work is extracted. Each time the particle collides with the walls, it exchanges energy with the thermal reservoir, $\rr$. As such, the source of work is the heat drawn from the reservoir.  }\label{Szilard engine}
\end{figure}

By considering an infinite ensemble of such boxes, the average state of the particle can be interpreted as being an ideal gas occupying volume $V_x$ for $x\in \{L,R\}$ which, after feedback, ``expands'' to volume $V$. If the box is in thermal contact with a single reservoir $\rr$ of temperature $T$, and the gas  expands quasistatically,   the engine will extract $W_x = K_B T\int_{V_x}^V dV'/V' = K_B T \ln(V/V_x)$ units of work, where $K_B$ is Boltzmann's constant. This is of course an average quantity of work, taken over the infinite ensemble of boxes. Moreover, the source of the extracted work is the heat drawn from the thermal reservoir.  As the (average) state of the system at the start and end of the process is the same -- an ideal gas occupying volume $V$ -- the Szilard engine is in apparent violation of the Kelvin statement of the second law; it is a cyclically operating device, the sole effect of which is to absorb energy in the form of heat from a single thermal reservoir and to produce an equal amount of work \cite{Balian-macro-1}.

As shown by Penrose and Bennett \cite{Penrose-stat-mech, Bennett-Landauer-review, Bennett-Landauer-Notes}, one may salvage the second law by observing that the  demon is itself a physical entity, whose memory is altered by the measuring process. In order to make the engine cyclical the demon's memory must be returned  to its initial configuration, i.e.,  the demon's memory must be ``reset'' or ``erased''.  If the erasure process is conducted by means of an interaction with the same thermal reservoir, it will require an average work cost no less than the average extracted work, which is dissipated as heat to the reservoir \cite{Landauer, Landauer-information-physical, Reeb-Wolf-Landauer}; we may never win in the long run. 

In recent years, much attention has been paid to the interplay between quantum theory and thermodynamics \cite{Anders-thermo-review,Goold-thermo-review,Millen-thermo-review, Horodecki2013, Anders-Measurement-Thermodynamics, Lostaglio2015b,  Karen-extractable-work-correlations, Gogolin2015a, Guryanova2015, YungerHalpern2015a, Alhambra2016a}. This has included the extension of work extraction through feedback control to the quantum regime,  culminating in both theoretical \cite{Zurek-Szilard, Dahlsten-Szilard,Sagawa-Szilard,Sagawa-feedback-control-Quantum, Jacobs-Feedback-freeenergy, Sagawa-Heat-engine-QI} and experimental \cite{Camati-Maxwell-Demon, Janet-Maxwell-Experiment} investigations.   Of particular interest to our discussion is the work presented  in \cite{Alexia-thermo-Measurement, Alexia-Maxwell-Measurement}, wherein the authors consider the possibility of a Maxwell demon engine that functions in thermal isolation. Here, the source of work can no longer be identified as heat from a thermal reservoir, but rather as the energetic changes due to projective measurements. Such quantum measurements, however,  ultimately  result from a physical interaction between the system to be measured, and the measuring apparatus; in the case of a Szilard engine, the measuring apparatus is the demon's memory. It stands to reason, therefore, that energetic considerations come to bear on the measuring process  \cite{Sagawa-thermodynamic-measurement, Kurt-measurement-thermo, Popescu-energy-conversation-measurement, Miyadera-Time-Energy-Measurement, Abdelkhalek-measurement}, which will pose limitations on the performance of Szilard engines that, in lieu of a thermal reservoir, draw power from projective measurements.

We recall from the classical Szilard engine that hidden entropy sinks, when the demon's memory is not explicitly  accounted for, allow for a violation of the second law. Similarly,   hidden work sources involved in the measuring process can also allow us to ``cheat''. Consequently, a constraint of primary importance that must be imposed on the measuring process of a Szilard engine is energy conservation; if the energy of the system is increased by projective measurements, the demon's energy must decrease in kind.   A central result from quantum measurement theory that is relevant to us is the Wigner-Araki-Yanase theorem \cite{Wigner-Measurement-conservation,Araki-Yanase, Miyadera-WAY-distinguishability, Loveridge-WAY,Loveridge-WAY-2,Mehdi-WAY} which, under additive conservation laws, will limit the observables that can be measured.   Using this, we shall show that while a Szilard engine can  be powered by projective measurements instead of heat from a reservoir, it will have to give up at least one of three features that are present in the classical Szilard engine. The three features of the classical Szilard engine in question are:

\

\begin{feat}\label{feature repeatable measurement}
The measurement is repeatable. If the demon measures the box and finds that the particle was on the right (left) hand side, a subsequent measurement would reveal that the particle is on the right (left) hand side with certainty. This allows for the interpretation that, after the measurement has been completed, the system ``possesses'' the revealed value.
\end{feat}
\begin{feat}\label{feature same entropy}
The weight's entropy does not change as a result of work extraction.  Work is extracted by raising the weight, thus increasing its gravitational potential energy. In general, the height of the weight's center of mass will be a fluctuating quantity, with an uncertainty $\Delta h$. However, $\Delta h$ does not change as a result of work extraction. In other words, the weight is neither ``cooled'' nor ``heated'' as it is elevated.
\end{feat}
\begin{feat}\label{feature positive work extraction}
The engine works reliably -- the work extracted is strictly positive for all measurement outcomes.  Whether the particle is on the right or left hand side of the box, the extracted work has the value $ W_x = K_B T \ln(V/V_x)$ where $x\in \{L,R\}$. As $V$ and $V_x<V$ are always positive, finite numbers, then  $W_x >0$ for all $x \in \{L,R\}$.
\end{feat}

\section{Modeling a quantum Szilard engine}
A general quantum Szilard engine is constituted of four subsystems: a system $\s$; a demon $\dd$; a weight $\ww$; and a thermal reservoir $\rr$. These have the Hilbert space $\h = \h\sub{\ww}\otimes \h\sub{\s}\otimes \h\sub{\dd}\otimes \h\sub{\rr}$, and respectively the Hamiltonians $H\sub{\ww}$, $H\sub{\s}$, $H\sub{\dd}$, and $H\sub{\rr}$. When describing operators that act non-trivially on only one subsystem, we shall omit identities on the other subsystems for simplicity. Furthermore, we shall only consider finite-dimensional Hilbert spaces.  This model has in common with \cite{thermo-individual-quantum,Aberg-Catalytic-Coherence} and \cite{Sagawa-Heat-engine-QI,Abdelkhalek-measurement} that it  includes respectively the weight and the demon's memory within the quantum description.  As with the classical Szilard engine, each cycle of our quantum Szilard engine involves two stages, namely, measurement and feedback.   Before $\dd$ can perform measurements in the next cycle, its memory must first be erased. This is achieved by an appropriate interaction with $\rr$. As the state of $\s$ can be different at the end of the cycle, then unlike the classical Szilard engine, the quantum Szilard engine is, strictly speaking, not cyclical. However, as will be shown, such non-cyclicality will not result in a violation of the second law.

All Szilard engines must satisfy the following two requirements. Here, we shall state them colloquially, but will offer mathematically precise formulations in the next two subsections.

\

\begin{req}\label{requirement energy}
Both the measuring and feedback processes must be energy conserving on the total system.
\end{req}
This is necessary for all work sources to be explicitly accounted for; if either the measuring or feedback process does not conserve the energy of the total system, then it will require work from an outside source. 

\

\begin{req}\label{requirement feedback}
If the demon's memory is in a state corresponding to a measurement outcome $x$, the feedback process must result in a closed evolution of  the compound of system plus weight (and reservoir, if it is present). After feedback, the demon's memory must remain in the same state. 
\end{req}
This is necessary in order to conform with the functioning of the classical Szilard engine described above. There, upon discovering the particle's location, the demon  arranges the weight-and-pulley mechanism accordingly so as to facilitate work extraction. After making its arrangements, the weight, system, and reservoir evolve as a closed, mechanically isolated system, while the demon's memory is unaltered. 

In the subsequent sections, we shall depart from the traditional set-up of the Szilard engine by altering the feedback stage; this will no longer involve $\rr$, and the source of work will not be identified as heat from the reservoir, but rather the internal energy of the compound $\s+\dd$. Each cycle of work extraction is depicted schematically in \fig{work extraction model}. Our work is similar in spirit to that of \cite{Alexia-Maxwell-Measurement}, except that we model both the weight and demon's memory as explicit quantum systems, and impose energy conservation on the measuring process.

\begin{figure}[!tb]
\includegraphics[width= 8 cm]{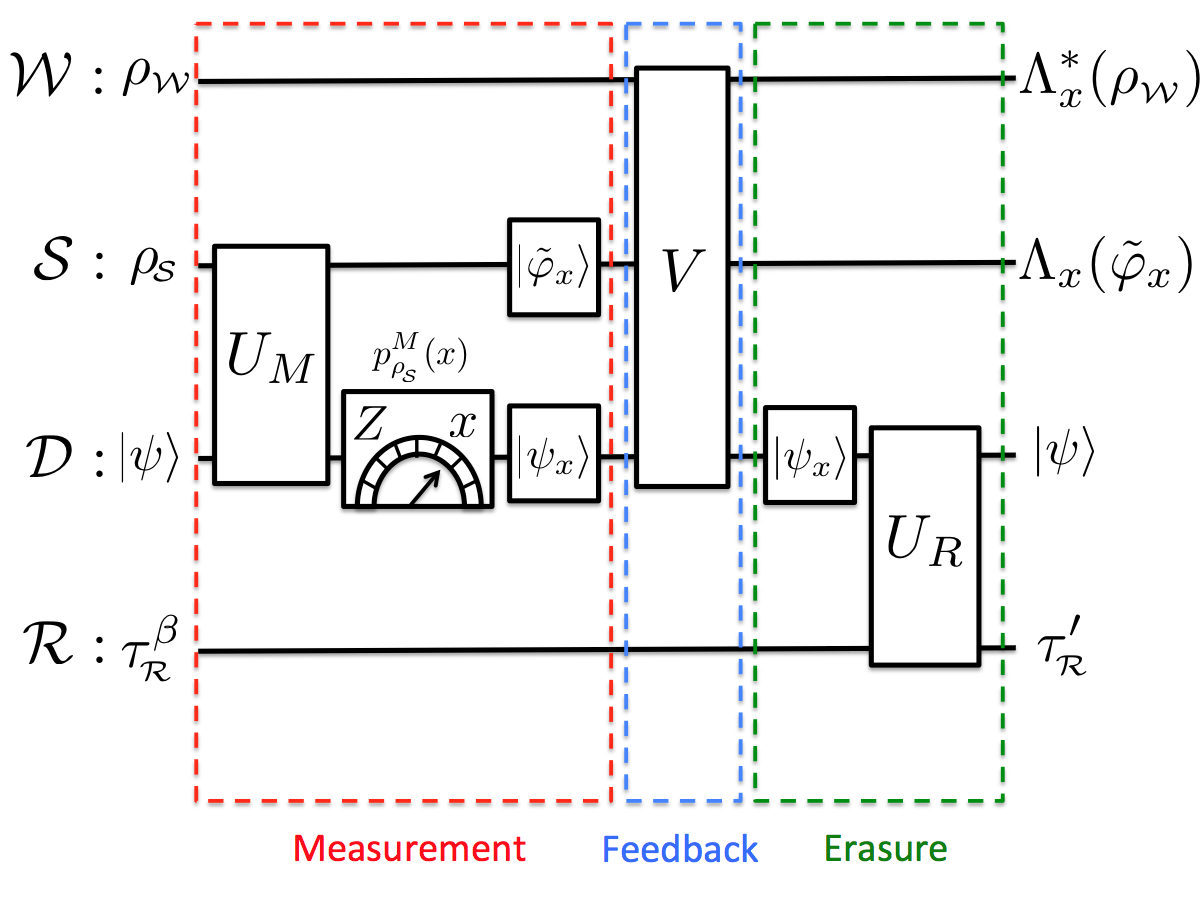}
\caption{The circuit model of measurement-assisted work extraction, without heat from the thermal reservoir. $\ww$, $\s$, $\dd$ and $\rr$ are initially prepared in states $\rho\subw{\ww}$, $\rho\sub{\s}$, $\ket{\psi}$, and $\tau\sub{\rr}^\beta$ respectively. (I) Measurement: First,  $\s$ and $\dd$ are coupled by the joint premeasurement unitary $U_M$. This correlates the two systems so that the measurement outcomes of the observable $M\sub{\s}$ on $\s$, namely $x\in\xx$, are ``stored'' in $\dd$'s memory as the eigenstates of an observable $Z\sub{\dd}$. These are the states $\{\ket{\psi_x}\}_{x\in \xx}$. The system's post-measurement state $\ket{\tilde \varphi_x}$ will be classically correlated with the demon memory state $\ket{\psi_x}$, occurring with a probability $p^M\sub{\rho_{\!_\s}}(x)$. (II) Feedback: The global feedback unitary operator $V$ then couples $\ww$ and $\s$ such that, conditional on the outcome $x$, they evolve by the CPTP maps $\Lambda_x^*$ and $\Lambda_x$, respectively. (III) Erasure: At the end of feedback, the demon's memory is erased by coupling to the thermal reservoir $\rr$ with the unitary interaction $U_R$.  }
\label{work extraction model}
\end{figure}

\subsection{Measurement stage}
During the measurement stage, the demon $\dd$ performs a measurement on $\s$, and by doing so prepares it in a state that is correlated with the measurement outcome. For now, we will restrict ourselves to standard, non-degenerate projective measurements, and shall generalise to degenerate observables in \app{degenerate observable section}.   If $\h\sub{\s} \simeq \co^d$, the observable can be represented as the self-adjoint operator 
\begin{align}\label{system observable}
M\sub{\s} = \sum_{x\in \xx} x\prs{\varphi_x},
\end{align} 
where $\xx := \{1,\dots,d\}$ are the measurement outcomes. Here $\prs{\varphi_x}\equiv |\varphi_x\>\<\varphi_x|$ is a  projection on the vector $\ket{\varphi_x} \in \h_\s$. We wish to model the measurement of $M\sub{\s}$ as resulting from a physical interaction between $\s$ and $\dd$, so that the outcomes $\xx$ are stored in the memory of $\dd$ by the orthogonal set of states $\{\ket{\psi_x}\in \h\sub{\dd}\}_{x\in\xx}$.  Therefore, we describe the measurement model  of $M\sub{\s}$, as defined in \eq{system observable}, by the tuple  $\mathscr{M}:= ( \h\sub{\dd}, \ket{\psi}, U_M, Z\sub{\dd})$ \cite{von-Neumann,Busch-operational,Busch-measurement,Busch-measurement-2,Heinosaari}. Here $\ket{\psi} \in \h\sub{\dd}$ is the initial state of $\dd$; $U_M$ is the \emph{premeasurement} unitary interaction between $\s$ and $\dd$, characterised by 
\begin{equation}\label{standard premeasurement unitary}
U_M: \ket{\varphi_x} \otimes\ket{\psi} \mapsto \ket{\tilde \varphi_x} \otimes \ket{\psi_x}
\end{equation}
where $\{\ket{\tilde \varphi_x}\}_{x\in\xx}$ can be any set of vectors on $\h\sub{\s}$, which do not have to be orthogonal; and 
\begin{align}\label{demon observable}
Z\sub{\dd} =  \sum_{x \in \xx} x P\sub{\dd}^x
\end{align}  
is an observable on $\dd$ with each outcome $x$ corresponding to the same for $M\sub{\s}$. Here, $P\sub{\dd}^x$ is a projection operator of arbitrary rank, such that for all $x\in \xx$, $\ket{\psi_x} \in P\sub{\dd}^x(\h\sub{\dd})$. If $\h\sub{\dd}\simeq\h\sub{\s}$, then $P\sub{\dd}^x= \prdd{\psi_x}$.  

For an arbitrary initial state $\rho\sub{\s}$ of $\s$, the total state of $\s+\dd$ after premeasurement is 
\begin{align}\label{premeasurement state}
\rho\sub{\s+\dd}^M:= U_M(\rho\sub{\s} \otimes \prdd{\psi})U_M^\dagger.
\end{align}
In order for the measuring process  to leave a classical record of outcomes, the demon's memory  must be  \emph{objectified} \cite{Mittelstaedt-measurement}. That is to say, after coupling $\s$ with $\dd$ by the premeasurement unitary as defined by \eq{standard premeasurement unitary}, thus preparing the  entangled state $\rho\sub{\s+\dd}^M$ as defined in \eq{premeasurement state}, we must prepare the statistical mixture
\begin{align}\label{Gemenge}
\rho\sub{\s+\dd}^{M,O} &:= \sum_{x \in \xx}P\sub{\dd}^x\rho\sub{\s+\dd}^M P\sub{\dd}^x,
\nonumber \\ &=  \sum_{x\in \xx} p_{\rho_{\!_\s}}^M(x) \prs{\tilde \varphi_x}\otimes \prdd{\psi_x},
\end{align}
where 
\begin{align}\label{Born rule}
p^M_{\rho_{\!_\s}}(x) := \tr[\prs{\varphi_x} \rho\sub{\s}]
\end{align} 
is the Born rule probability of observing outcome $x$, given a measurement of $M\sub{\s}$ on $\s$, prepared in state $\rho\sub{\s}$. \eq{Gemenge} is a proper mixture, or a \emph{Gemenge} $(p_{\rho_{\!_\s}}^M(x) ,  \prs{\tilde \varphi_x}\otimes \prdd{\psi_x})$, which can be interpreted as each state $\prs{\tilde \varphi_x}\otimes \prdd{\psi_x}$ being prepared according to a probability distribution $p_{\rho_{\!_\s}}^M(x)$, as given by \eq{Born rule}. Moreover, $\{\ket{\tilde \varphi_x}\}_{x\in\xx}$ can be interpreted as  the set of post-measurement states on $\s$. We may objectify $\dd$ by performing an unselective L\"uders measurement of $Z\sub{\dd}$ \cite{Busch-measurement-2}, as defined in \eq{demon observable}, on $\dd$. Alternatively, as shown in \cite{Abdelkhalek-measurement}, $\dd$ can be objectified by unitarily coupling it with an auxiliary system.  In the subsequent section we show that imposing \reqref{requirement feedback} on the feedback process implies that it does not matter whether we objectify the demon before or after the feedback stage.

\

\begin{defn}\label{definition energy conserving}
Consider a system with Hilbert space $\h$ and Hamiltonian $H$. The completely positive, trace preserving (CPTP) map $\ee$ is said to conserve energy if 
\begin{align}
\tr[H \rho] = \tr[H \ee(\rho)]
\end{align}
for all states $\rho$ on $\h$. 
\end{defn}

\

\begin{lem}\label{lemma measurement energy conserving}
The measuring process satisfies \reqref{requirement energy}, i.e., is energy conserving, if both $[Z\sub{\dd}, H\sub{\dd}]_-=\zero$ and $[U_M, H\sub{\s}+ H\sub{\dd}]_-=\zero$, where $H\sub{\s}$ and $H\sub{\dd}$ are the system and demon Hamiltonians, respectively, and $Z\sub{\dd}$ is the demon observable defined in \eq{demon observable}.
\end{lem}
\begin{proof}
The measuring process consists of premeasurement and objectification. Given \defref{definition energy conserving}, these are energy conserving if
\begin{align}
\tr[(H\sub{\s} + H\sub{\dd})\rho\sub{\s+\dd}^{M,O}] = \tr[(H\sub{\s} + H\sub{\dd})\rho\sub{\s}\otimes \prdd{\psi}]
\end{align}
for all $\rho\sub{\s}$ on $\h\sub{\s}$, where $\rho\sub{\s+\dd}^{M,O}$ is given by \eq{Gemenge}.  Therefore, we must have $[U_M, H\sub{\s}+ H\sub{\dd}]_-=\zero$ and $[P\sub{\dd}^x, H\sub{\dd}]_-=\zero$ for all $x\in \xx$. The latter condition is equivalent to $[Z\sub{\dd}, H\sub{\dd}]_-=\zero$.  
\end{proof}

Now we may analyse \featref{feature repeatable measurement} with respect to \reqref{requirement energy}. 

\

\begin{lem}\label{Repeatability lemma}
Let the measuring process satisfy \reqref{requirement energy}. It follows that the measurement of $M\sub{\s}$, as defined by \eq{system observable}, will satisfy \featref{feature repeatable measurement}, i.e, it will be repeatable, if and only if the post-measurement states $\{\ket{\tilde \varphi_x}\}_{x\in \xx}$ are eigenvectors of $H\sub{\s}$.
\end{lem}
\begin{proof}
The post-measurement state of $\s$, conditional on outcome $x$, is $\ket{\tilde \varphi_x}$. The probability of observing outcome $x$ in a subsequent measurement of $M\sub{\s}$ will be $p^M_{\tilde \varphi_x}(x) = |\<\tilde \varphi_x|\varphi_x\>|^2$, as determined by \eq{Born rule}. This equals unity if and only if $\ket{\tilde \varphi_x} = e^{i\theta}\ket{\varphi_x}$. Therefore, $\{\ket{\tilde \varphi_x}\}_{x\in \xx}$ must be eigenvectors of $M\sub{\s}$. 

To show that $\{\ket{\tilde \varphi_x}\}_{x\in \xx}$ must be eigenvectors of $H\sub{\s}$ if the measurement is repeatable, we use the WAY theorem. The WAY theorem can be stated thusly: let the premeasurement unitary operator in the measurement model of $M\sub{\s}$, i.e., $U_M$, commute with $H\sub{\s} + H\sub{\dd}$. If the measurement of $M\sub{\s}$ is repeatable, or $[Z\sub{\dd},H\sub{\dd}]_-=\zero$, where $Z\sub{\dd}$ is defined in \eq{demon observable}, then $[M\sub{\s}, H\sub{\s}]_-=\zero$. We refer to \cite{Loveridge-WAY} for a proof. If $M\sub{\s}$ commutes with $H\sub{\s}$, then they will share the same eigenvectors. 
\end{proof}

\subsection{Feedback stage}

During the feedback stage, the demon brings the system in contact with the weight, $\ww$, which is initially prepared in state $\rho\subw{\ww}$. Conforming with \reqref{requirement feedback}, the demon then evolves the compound system of $\ww+\s$  by the unitary operator  $U_x$, which is chosen conditional on the measurement outcome $x\in \xx$. We wish to determine the global feedback unitary operator $V$ that achieves this.

\

\begin{lem}\label{feedback lemma memory}
Feedback is implemented by a unitary operator $V$ acting on the compound system $\ww+\s+\dd$. $V$ will satisfy \reqref{requirement feedback} if and only if it can be written as 
\begin{align}\label{feedback unitary}
V = \sum_{x\in \xx} U_x \otimes P\sub{\dd}^x,
\end{align}
such that $U_x$ are unitary operators on $\h\sub{\ww}\otimes \h\sub{\s}$, and $P\sub{\dd}^x$ are the projection operators defined in \eq{demon observable}. 
\end{lem}
\begin{proof}
\reqref{requirement feedback} states that if the demon is in a state corresponding to a measurement outcome $x$, the system and weight must undergo a closed evolution. Consequently,  $V$ must satisfy
\begin{align}\label{feedback memory equation 1}
V(\ket{\Psi}\otimes \ket{\psi_x}) = (U_x\ket{\Psi})\otimes \ket{\psi_x}
\end{align}
for all $x\in \xx$ and $\ket{\Psi}\in \h\sub{\ww}\otimes \h\sub{\s}$, where $\ket{\psi_x}$ is an eigenstate of the demon observable $Z\sub{\dd}$ as defined in \eq{demon observable}. This is clearly satisfied if $V$ is of the form \eq{feedback unitary}. To prove only if, we note that  \eq{feedback memory equation 1} implies that
\begin{align}\label{eq:Req2onlyif}
V(\ket{\Psi}\otimes \ket{\psi_x}) = (P\sub{\dd}^x V P\sub{\dd}^x)(\ket{\Psi}\otimes \ket{\psi_x})
\end{align}
for all $x\in \xx$, where $P\sub{\dd}^x$ is a projection on the subspace of $\h_\dd$ that contains $\ket{\psi_x}$. Therefore, it follows that 
\begin{align}
V = \sum_{x\in\xx} P\sub{\dd}^x V P\sub{\dd}^x,
\end{align}
and so $V$ must be of the form \eq{feedback unitary}.
\end{proof}
\begin{cor}\label{feedback objectification corollary}
Let the feedback unitary satisfy \reqref{requirement feedback}. Then the state of the compound $\ww+\s+\dd$ will be identical whether  $\dd$ is objectified prior to feedback, or after it.
\end{cor}
\begin{proof}
The compound of $\s+\dd$ after premeasurement and objectification is given by  \eq{Gemenge}. After feedback, the state of the compound $\ww+\s+\dd$ is 
\begin{align}
&V(\rho\subw{\ww}\otimes\rho\sub{\s+\dd}^{M,O})V^\dagger \nonumber \\ 
& \, \, \,= V\left(\sum_{x \in \xx}P\sub{\dd}^x(\rho\subw{\ww}\otimes\rho\sub{\s+\dd}^{M}) P\sub{\dd}^x\right)V^\dagger.
\end{align}
If the feedback unitary is of the form \eq{feedback unitary}, then $[V,  P\sub{\dd}^x]_-=\zero$ for all $x\in \xx$, and so we have
\begin{align}
&V\left(\sum_{x \in \xx}P\sub{\dd}^x(\rho\subw{\ww}\otimes\rho\sub{\s+\dd}^{M}) P\sub{\dd}^x\right)V^\dagger \nonumber \\
&=\sum_{x \in \xx}P\sub{\dd}^x V(\rho\subw{\ww}\otimes\rho\sub{\s+\dd}^{M})V^\dagger P\sub{\dd}^x.
\end{align}
The second line corresponds to performing feedback after premeasurement, but before objectification has occurred. 
\end{proof}

We now show that if $V$ as defined by \eq{feedback unitary} is to satisfy \reqref{requirement energy}, then each $U_x$ must conserve $H\sub{\ww}+ H\sub{\s}$. 

\

\begin{lem}\label{feedback lemma energy}
Let $V$ be a feedback unitary operator that satisfies \reqref{requirement feedback}. It will also satisfy \reqref{requirement  energy} if and only if: (i)  $[U_x, H\sub{\ww} + H\sub{\s}]_-=\zero$ for all $x\in \xx$; and (ii) for every subset $\xx' \subseteq \xx$ such that $U_x=U_y$ for all $x,y \in \xx'$, $\sum_{x\in \xx'} [P\sub{\dd}^x, H\sub{\dd}]_- = \zero$.  
\end{lem}
\begin{proof}
In order for $V$ as defined by \eq{feedback unitary} to conserve the total energy, by \defref{definition energy conserving} we require that 
\begin{align}
&\tr[ H V \rho V^\dagger] = \tr[H \rho]\label{Work extraction energy condition}
\end{align}
for all states $\rho$ on $\h\sub{\ww}\otimes \h\sub{\s}\otimes \h\sub{\dd}$, where $H = H\sub{\ww} + H\sub{\s} + H\sub{\dd}$. Therefore, $V$ must commute with the total Hamiltonian. Because of the additivity of the Hamiltonian, $[V, H]_- = \zero$ can be written as
\begin{align}\label{feedback commutation}
\sum_{x\in\xx} [U_x, H\sub{\ww} + H\sub{\s}]_- \otimes P\sub{\dd}^x = - \sum_{x\in \xx} U_x\otimes [P\sub{\dd}^x, H\sub{\dd}]_-.
\end{align}
Given an arbitrary pair of states $\ket{\psi_x}\in P^x\sub{\dd}(\h\sub{\dd})$ and $\ket{\psi_y}\in P^y\sub{\dd}(\h\sub{\dd})$, such that $x\ne y$, and referring to the right hand and left hand sides of \eq{feedback commutation} as RHS and LHS, respectively, we see that 
\begin{align}
\<\psi_x| \mathrm{LHS} |\psi_y\> &= \zero, \nonumber \\ 
\<\psi_x| \mathrm{RHS} |\psi_y\> &= \<\psi_x| H\sub{\dd}|\psi_y\> (U_y - U_x).
\end{align}
However, given \eq{feedback commutation}, we must have $\<\psi_x| \mathrm{LHS} |\psi_y\> =  \<\psi_x| \mathrm{RHS} |\psi_y\>$.  This is satisfied if either:  (i) $[P\sub{\dd}^z,H\sub{\dd}]_-=\zero$ for $z\in \{x,y\}$; or (ii) $U_x = U_y$. Option (i) satisfies the if statement of the Lemma. Option (ii) implies that \eq{feedback commutation} is satisfied if
\begin{align}\label{feedback commutation equal unitary}
 [U_{\xx'}, H\sub{\ww} + H\sub{\s}]_- \otimes P\sub{\dd}^{\xx'} = - U_{\xx'}\otimes [P\sub{\dd}^{\xx'}, H\sub{\dd}]_-
\end{align}
for all maximal subsets $\xx' \subseteq \xx$ such that, given all $x,y \in \xx'$, $U_x = U_y = U_{\xx'}$. Here we define $P\sub{\dd}^{\xx'}:= \sum_{x\in \xx'} P\sub{\dd}^x$. 

\eq{feedback commutation equal unitary} is satisfied if : (a) $[P\sub{\dd}^{\xx'}, H\sub{\dd}]_- \propto  P\sub{\dd}^{\xx'}$ and $[U_{\xx'}, H\sub{\ww} + H\sub{\s}]_- \propto U_{\xx'}$; or (b) if $[P\sub{\dd}^{\xx'}, H\sub{\dd}]_- =\zero$ and $[U_{\xx'}, H\sub{\ww} + H\sub{\s}]_- =\zero$. It is easy to verify that (a) is impossible, and so only option (b) is available. This concludes the proof of the only if portion of the Lemma.  
\end{proof}

For each measurement outcome $x$, as a result of the global feedback unitary operator $V$ given in \eq{feedback unitary}, $\s$ and $\ww$ undergo the complementary CPTP maps
\begin{align}\label{feedback CPTP maps}
\Lambda_x&: \prs{\tilde \varphi_x} \mapsto \tr\sub{\ww}[U_x(\rho\subw{\ww}\otimes\prs{\tilde \varphi_x})U_x^\dagger], \nonumber \\
\Lambda_x^* &: \rho\subw{\ww} \mapsto \tr\sub{\s}[U_x(\rho\subw{\ww}\otimes\prs{\tilde \varphi_x})U_x^\dagger],
\end{align}
where we recall that $\{\ket{\tilde \varphi_x}\}_{x\in\xx}$ are the post-measurement states of $\s$. 

We now wish to define the (average) work that is transferred from $\s$ into $\ww$, for each measurement outcome, as a result of feedback. To this end, we use the following definition.

\

\begin{defn}\label{definition work}
For each measurement outcome $x\in \xx$, the average work transferred into the weight is defined as 
\begin{align}\label{equation definition work}
W_x := F(\Lambda_x^*[\rho\subw{\ww}]) - F(\rho\subw{\ww}),
\end{align}
where: $\Lambda_x^*$ is the CPTP map defined by \eq{feedback CPTP maps};
\begin{align}\label{free energy}
F(\rho) := \tr[H \rho] - K_B T \, S(\rho)
\end{align}
is the non-equilibrium free energy of a system with state $\rho$, relative to the Hamiltonian $H$ and temperature $T$; and $S(\rho):= -\tr[\rho \ln(\rho)]$ is the von-Neumann entropy of $\rho$. 
\end{defn}
This definition has been argued for previously in \cite{Gemmer2015,Gallego2015-Work-definition}. Even though the thermal reservoir is not involved during feedback, it is still part of the thermodynamic context of the Szilard engine. As such, work can be extracted from both the system, and the weight, by letting them interact appropriately with the reservoir. Therefore, the quantifier of work transfer must be temperature dependent, in the form of free energy difference, in order to : (i) ensure consistency with the ``internal'' description of work extraction from $\s$, wherein the weight is not included in the quantum description; and (ii) avoid violation of the second law.     For a detailed argument we refer the reader to \app{appendix work definition}. We note that an alternative definition for work transfer to the weight is the increase in the internal energy of $\ww$. While this formulation will be consistent with the second law only if the feedback unitary $V$ induces unital dynamics on the system $\s$ \cite{Morikuni2017}, \defref{definition work} does not suffer from such limitations. Moreover, \defref{definition work} reduces to the increase in internal energy  when \featref{feature same entropy} is satisfied.

Now that we have defined work extraction, we may analyse this with respect to \featref{feature same entropy}.

\

\begin{defn}\label{definition entropy increase}
The Szilard engine satisfies \featref{feature same entropy} if for all $x\in \xx$, 
\begin{align}
S(\Lambda_x^*[\rho\subw{\ww}]) = S(\rho\subw{\ww}). 
\end{align}
\end{defn}

\begin{lem}\label{Feature 2 lemma}
When the Szilard engine satisfies \featref{feature same entropy}, it follows that 
\begin{align}\label{feature 2 lemma equation}
W_x \leqslant \<\tilde \varphi_x|H\sub{\s}|\tilde \varphi_x\> - \min[\sigma(H\sub{\s})].
\end{align}
where $\sigma(H\sub{\s})$ is the spectrum of $H\sub{\s}$.  
\end{lem}
\begin{proof}
The work transferred into $\ww$ is, by \defref{definition work} and \lemref{feedback lemma energy}, given as 
\begin{align}\label{conditional work weight}
W_x&:=  \tr[H\sub{\s}(\prs{\tilde\varphi_x} - \Lambda_x[\tilde\varphi_x])] \nonumber \\ &\, \, \, \, + K_B T \,(S(\rho\subw{\ww}) - S(\Lambda_x^*[\rho\subw{\ww}])).
\end{align}
As $\tr[H\sub{\s} \Lambda_x[\tilde\varphi_x]]\geqslant \min[\sigma(H\sub{\s})]$, it follows that 
\begin{align}\label{weight entropy increase detrimental}
W_x&\leqslant  \<\tilde \varphi_x|H\sub{\s}|\tilde \varphi_x\> - \min[\sigma(H\sub{\s})] \nonumber \\ &\, \, \, \, + K_B T \,(S(\rho\subw{\ww}) - S(\Lambda_x^*[\rho\subw{\ww}])).
\end{align}
If the Szilard engine satisfies \featref{feature same entropy}, then by \defref{definition entropy increase} we have \eq{feature 2 lemma equation}.
\end{proof}

\section{The impossibility theorem}

\begin{figure}[!b]
\includegraphics[width= 8 cm]{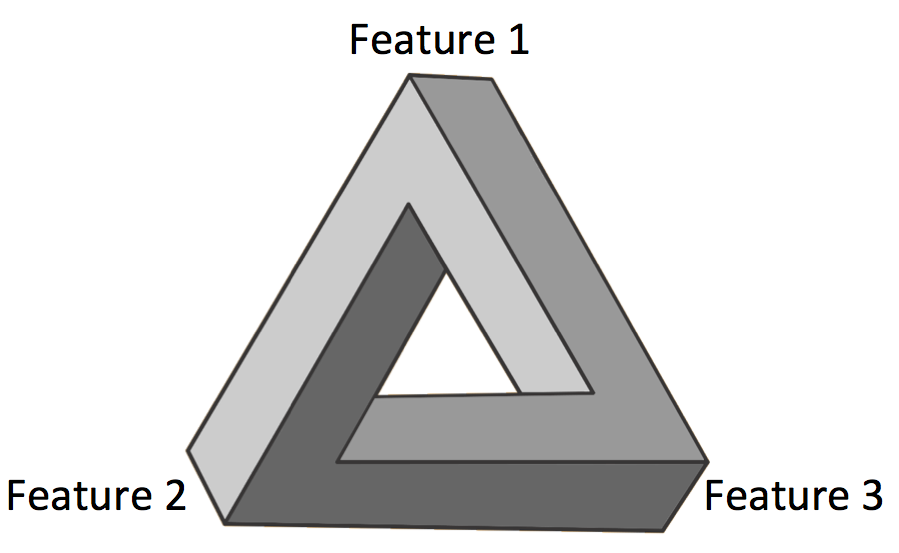}
\caption{ The impossible triangle of a quantum Szilard engine powered by projective measurements. Features 1, 2, and 3 signify respectively the repeatability of the measurement; invariant weight entropy; and the reliability of the engine. The fact that only two vertices of the impossible triangle can be physically connected, but not the third, represents the result that all three features cannot be simultaneously satisfied. }\label{Nogo-triangle}
\end{figure}

We are now ready to prove a main result of this paper. The impossibility theorem is  illustrated  by  Penrose's impossible triangle in \fig{Nogo-triangle}. 

\

\begin{thm}\label{impossibility theorem}
Consider a quantum Szilard engine that, during the feedback stage, operates in thermal isolation. Let the engine satisfy \reqref{requirement energy} and \reqref{requirement feedback}. It follows that if the engine satisfies any two from \featref{feature repeatable measurement}, \featref{feature same entropy}, and \featref{feature positive work extraction}, it will necessarily fail to satisfy the third.
\end{thm}
\begin{proof}
Let the engine satisfy \featref{feature repeatable measurement} and \featref{feature same entropy}. By \lemref{Repeatability lemma} the post-measurement states $\{\ket{\tilde \varphi_x}\}_{x\in \xx}$ are the eigenvectors of $M\sub{\s}$ and, hence, $H\sub{\s}$. Consequently, for some outcome $x\in \xx$, $\<\tilde \varphi_x|H\sub{\s}|\tilde \varphi_x\> = \min[\sigma(H\sub{\s})]$. By \lemref{Feature 2 lemma}, for this outcome we have $W_x\leqslant 0$, and \featref{feature positive work extraction} cannot be satisfied. 

Let the engine satisfy \featref{feature repeatable measurement} and \featref{feature positive work extraction}. By \lemref{Repeatability lemma} the post-measurement states $\{\ket{\tilde \varphi_x}\}_{x\in \xx}$ are the eigenvectors of $M\sub{\s}$ and, hence, $H\sub{\s}$. Consequently, for some outcome $x\in \xx$, $\<\tilde \varphi_x|H\sub{\s}|\tilde \varphi_x\> = \min[\sigma(H\sub{\s})]$. By \lemref{Feature 2 lemma}, for this outcome $W_x>0$ only if $S(\Lambda_x^*[\rho\subw{\ww}]) < S(\rho\subw{\ww})$. Hence, \featref{feature same entropy} cannot be satisfied. 

Let the engine satisfy \featref{feature same entropy} and  \featref{feature positive work extraction}. By \lemref{Feature 2 lemma}, for all $x\in \xx$, the work is bounded as $W_x \leqslant \<\tilde \varphi_x|H\sub{\s}|\tilde \varphi_x\> - \min[\sigma(H\sub{\s})]$. As $W_x>0$ for all $x\in \xx$, it follows that $\<\tilde \varphi_x|H\sub{\s}|\tilde \varphi_x\> > \min[\sigma(H\sub{\s})]$ for all $x\in \xx$. Therefore, the post-measurement states $\{\ket{\tilde \varphi_x}\}_{x\in\xx}$ cannot be the eigenvectors of $H\sub{\s}$. By \lemref{Repeatability lemma}, \featref{feature repeatable measurement} cannot be satisfied.
\end{proof}

\thmref{impossibility theorem}, simply stated, says that if the system is measured with respect to a non-degenerate observable, in a repeatable and energy conserving fashion, it must be projected onto the eigenstates of $H\sub{\s}$. Consequently, if we do not allow the weight's entropy to decrease,   then for the outcome that projects the system onto the groundstate of $H\sub{\s}$, zero work can be extracted.

In \app{Qubit examples}, we illustrate the incompatibility between the three features by looking at a concrete model where both $\s$ and $\dd$ are qubits, while $\ww$ is a harmonic oscillator. In \app{circumventing theorem 1} we show that \thmref{impossibility theorem} can be circumvented if: (i) the thermal reservoir is involved during the feedback stage so that, just as in the classical Szilard engine, the source of work will be heat drawn from the reservoir; or (ii) the observable measured on $\s$ is degenerate and is measured ``inefficiently''.

\section{Net work extraction per cycle}

\fig{work extraction model} depicts a single cycle of the Szilard engine under consideration. In \app{Energetic contributions in one cycle of the quantum Szilard engine} we evaluate the net work extraction per cycle,  wherein we do not distinguish between measurement outcomes. Labeling the ``coarse-grained'' work transferred to the weight as $W_\xx:= F(\rho\sub{\ww}') - F(\rho\sub{\ww})$, and the work cost of erasure as $W_R$, the net coarse-grained work  is shown to  obey the inequality 
\begin{align}\label{net average work extraction main}
W^\mathrm{net}_\xx:= W_\xx - W_R &\leqslant     F(\rho\sub{\s}) - F(\rho\sub{\s}'),
\end{align}
where  $\rho\subw{\s}'$ and $\rho\subw{\ww}'$ are the average states of $\s$ and $\ww$ at the end of the cycle, respectively, obtained by sampling the states $\Lambda_x(\tilde \varphi_x)$ and $\Lambda_x^*(\rho\sub{\ww})$  by the probability distribution $p^M_{\rho_{\!_\s}}(x)$ as defined by \eq{Born rule}. We note that \eq{net average work extraction main} holds irrespective of whether the Szilard engine satisfies any of \featref{feature repeatable measurement}, \featref{feature same entropy}, or \featref{feature positive work extraction}.  Moreover, we note that the coarse-grained work is  generally smaller than the average work, i.e., $W_\xx \leqslant \<W_x\> := \sum_{x\in \xx} p^M_{\rho_{\!_\s}}(x) W_x$, where $W_x$ is defined in \eq{equation definition work}. While the coarse-grained work extraction obeys the second law, the average work will not; if $\rho\sub{\s}$ is thermal, then $W^\mathrm{net}_\xx \leqslant 0$ whereas $\<W_x^\mathrm{net}\>:= \<W_x\> - W_R$ can be positive. 

To be sure, the second law is a statistical statement, held true precisely when we do not have access to the individual measurement outcomes. Let us recall the definition for work transferred into the weight when it transforms as $\rho\sub{\ww} \mapsto \Lambda_x^*(\rho\sub{\ww})$,  given by \defref{definition work} and articulated in \app{appendix work definition}. This was given operational meaning as being the maximum value of  work that can be extracted from the weight, by an isothermal process $\Lambda_x^*(\rho\sub{\ww}) \mapsto \rho\sub{\ww}$ involving the reservoir of temperature $T$.  However, if we were to forget the measurement outcomes, then we could not use such information to tailor our process of extracting work from the weight. Indeed, this protocol must be designed with only the average state of the weight in mind. The maximum value of work extractable from the weight, given an isothermal process $\rho\sub{\ww}' \mapsto \rho\sub{\ww}$, is precisely  $W_\xx$.

\section{Discussion}
 We give  a general mathematical description of a quantum Szilard engine that  operates in two stages, namely, projective measurement and feedback. In our model, in contradistinction to the classical Szilard engine, the feedback stage does not involve the thermal reservoir.  Here, the source of work is the energetic changes due to (non-degenerate)  projective measurements. In order to avoid cheating by the inclusion of hidden work sources, we impose energy conservation on the measuring process. As a result of the Wigner-Araki-Yanase theorem, the observables that the demon can measure will be limited to those that commute with the system's Hamiltonian.

We showed that while the Szilard engine, in lieu of a thermal reservoir, can be powered by (non-degenerate) projective measurements, it cannot simultaneously satisfy three features of the classical Szilard engine model; the conjunction of any two will preclude the possibility of the third.   These features are: (i) the measurement performed by the demon is repeatable, meaning that conditional on obtaining outcome $x$, a subsequent measurement of the same observable would yield $x$ with certainty; (ii) the weight's entropy does not change as a result of feedback; and (iii) work extraction is reliable, i.e., is strictly positive for all measurement outcomes. This observation is a first step towards developing ``second-law-like'' relations in the context of measurement-assisted feedback control beyond thermality. While the second law results from entropic considerations, these ``second-law-like'' relations would result from energy conservation of unitary interactions that implement measurements.    

The Szilard engine here discussed is, strictly speaking, not cyclical; at the end of a cycle of work extraction, the state of the system, $\rho\sub{\s}'$, will not be the same as its initial state, $\rho\sub{\s}$. For the engine to be made cyclical, therefore, we must have at our disposal an infinite supply of systems with state $\rho\sub{\s}$ such that, at the end of each cycle, the system's state is swapped with one of these. One example of such ``free resources'' is if $\rho\sub{\s}$ is thermal. Here, we may interpret the closure of the cycle to result from the system being brought to thermal equilibrium with the reservoir.

The strict  non-cyclicality of the engine notwithstanding, the statistical second law will not be violated. This is because, when taking the erasure cost of the demon into consideration, the total net work extracted from the system will be bounded by the decrease in its free energy -- a quantity that will not be positive if the system is initially at thermal equilibrium. However, this requires a careful consideration of how one should evaluate work when choosing to ``forget'' the measurement outcomes -- precisely the domain where the second law is applicable. As with unselective measurements,  the work transferred to the weight when the indvidual measurement outcomes are not distinguished from one another must be defined by how the weight's state changes \emph{on average}. Indeed, the extractable work from the weight, when the measurement outcomes are forgotten, is smaller than the average value of work, when the measurement outcomes are taken into consideration.

\begin{acknowledgments}
The authors would like to thank L. D. Loveridge, K. Abdelkhalek, D. Reeb, K. Hovhannisyan, and H. Miller for the useful discussions that helped in developing the ideas presented in this paper. J. A. acknowledges support from EPSRC, grant EP/M009165/1, and the Royal Society. This research was supported by the COST network MP1209 ``Thermodynamics in the quantum regime''. 

\end{acknowledgments}

\makeatletter
\renewcommand\p@subsection{\thesection\,}
\makeatother
\makeatletter
\renewcommand\p@subsubsection{\thesection\,\thesubsection\,}
\makeatother

\appendix

\section{Definition of work transferred into the weight}\label{appendix work definition}

Here we wish to justify defining the work transferred into the weight, as a result of feedback, by \defref{definition work}. To this end, let us first recall a known result from standard non-equilibrium quantum thermodynamics. In the \emph{internal} description of work extraction, in contradistinction to the \emph{external} description, the weight is not included in the quantum formalism. Here, the work extracted from a system undergoing a (non-energy conserving) unitary evolution is defined as the decrease in its internal energy. Consequently, if a system $\s$ undergoes a transformation $\rho\sub{\s} \mapsto \rho\sub{\s}':= \tr\sub{\rr}[U(\rho\sub{\s}\otimes \tau^\beta\sub{\rr})U^\dagger]$, where $U$ is a global unitary operator and $\tau\sub{\rr}^\beta := e^{-\beta H_\rr}/\tr[e^{-\beta H_\rr}]$ is the thermal state of the thermal reservoir $\rr$, with $\beta := (K_B T)^{-1}$ the inverse temperature, the work extracted obeys the inequality 
\begin{align}\label{free energy work inequality}
W_\mathrm{ext}(\rho\sub{\s} \mapsto \rho\sub{\s}')&:= \tr[(H\sub{\s} + H\sub{\rr})\rho\sub{\s}\otimes \tau^\beta\sub{\rr}] \nonumber \\ 
& \, \, - \tr[(H\sub{\s} + H\sub{\rr})U(\rho\sub{\s}\otimes \tau^\beta\sub{\rr})U^\dagger]\nonumber \\ 
&\leqslant F(\rho\sub{\s}) - F(\rho\sub{\s}'),
\end{align}
with the equality obtained when the interaction between system and thermal reservoir is ``quasi-static'' \cite{Anders-discrete-thermo}.

Therefore, \defref{definition work} can be justified with the following argument. When the weight interacts with the system, thereby transforming as $\rho\subw{\ww} \mapsto \Lambda_x^*(\rho\subw{\ww})$, where $\Lambda_x^*$ is given by \eq{feedback CPTP maps}, work is \emph{transferred} to it. We may then perform the reverse transformation on the weight, i.e., $ \Lambda_x^*(\rho\subw{\ww}) \mapsto \rho\subw{\ww}$, by an appropriate unitary interaction with the thermal reservoir, so as to extract this work. The work extracted here will be in the internal description, as there is no second weight into which the work is being transferred. By \eq{free energy work inequality}, the work we may extract obeys the inequality
\begin{align}\label{work extraction from weight}
W_\mathrm{ext}(\Lambda_x^*(\rho\subw{\ww})  \mapsto \rho\subw{\ww}) \leqslant F(\Lambda_x^*[\rho\subw{\ww}]) - F(\rho\subw{\ww}).
\end{align}
Clearly, the work transferred into the weight must be at least as great as the work that can be extracted from the weight, i.e., 
\begin{align}
W_x \geqslant W_\mathrm{ext}(\Lambda_x^*(\rho\subw{\ww})  \mapsto \rho\subw{\ww}).
\end{align}
A natural assumption to make is that, since the process of transferring work into the weight is independent of the process by which work is extracted from the weight, the right hand side of the above equation should be replaced by the upper bound of \eq{work extraction from weight}. If we also take the view that transferring more work into the weight than can possibly be extracted from it is physically meaningless, we arrive at \defref{definition work}. 

We also note that \defref{definition work} is consistent with the internal description of work from the system $\s$, and that it satisfies the second law. 

\

\begin{lem}
Let the system and weight be initially prepared in the states $\rho\sub{\s}$ and $\rho\subw{\ww}$, respectively. Let the two systems evolve by a unitary operator $U$ that conserves the total Hamiltonian $H\sub{\ww} + H\sub{\s}$, and induces the complementary CPTP maps  $\Lambda$ on $\s$ and $\Lambda^*$ on $\ww$. Then the work transferred into the weight, $W$, as defined by \defref{definition work}, will never exceed the maximum work that can be directly extracted from the system by the process $\rho\sub{\s} \mapsto \Lambda(\rho\sub{\s})$, in the internal description, and using a single thermal reservoir at temperature $T$. If $\rho\sub{\s}$ is thermal, then $W$ cannot be positive. 
\end{lem}
\begin{proof}
By \defref{definition work}, energy conservation of $U$, and the subadditivity of the von-Neumann entropy, we have 
\begin{align}
W &:= F(\Lambda^*[\rho\subw{\ww}]) - F(\rho\subw{\ww}),\nonumber \\
& = \tr[H\sub{\ww} (\Lambda^*[\rho\subw{\ww}] - \rho\subw{\ww})]  \nonumber \\ 
 &\, \, \, \, \, + K_B T \, (S(\varrho\subw{\ww}) - S(\Lambda^*[\rho\subw{\ww}])), \nonumber \\
& = \tr[H\sub{\s} ( \rho\sub{\s} - \Lambda[\rho\sub{\s}] )] \nonumber \\ 
 &\, \, \, \, \, + K_B T \, (S(\varrho\subw{\ww}) - S(\Lambda^*[\rho\subw{\ww}])), \nonumber \\
 &\leqslant \tr[H\sub{\s} ( \rho\sub{\s} - \Lambda[\rho\sub{\s}])] \nonumber \\ 
 &\, \, \, \, \, + K_B T \, (S(\Lambda[\rho\sub{\s}]) - S(\rho\sub{\s})), \nonumber \\
 &= F(\rho\sub{\s}) - F(\Lambda[\rho\sub{\s}]).
\end{align}
By \eq{free energy work inequality}, we see that $W$ is never greater than the upper bound of $W_\mathrm{ext}(\rho\sub{\s} \mapsto \Lambda[\rho\sub{\s}])$. Moreover, if the system is initially in the thermal state $\rho\sub{\s} = \rho\sub{\s}^\beta := e^{-\beta H\sub{\s}}/\tr[e^{-\beta H\sub{\s}}]$, we have
\begin{align}
W &\leqslant F(\rho\sub{\s}^\beta) - F(\Lambda[\rho\sub{\s}^\beta]), \nonumber \\
&=  -K_B T \,  S(\Lambda[\rho\sub{\s}^\beta]\| \rho\sub{\s}^\beta),
\end{align}
where $S(\rho\| \sigma):= \tr[\rho(\ln(\rho) - \ln(\sigma))]$ is the entropy of $\rho$ relative to $\sigma$, which is a non-negative number and vanishes if and only if $\rho = \sigma$. Therefore, $W \leqslant 0$. 
\end{proof}

\section{An example with qubits}\label{Qubit examples}
As an illustrative example, consider the simple case where $\s$ and $\dd$ are both qubits, with the Hamiltonians 
\begin{align}
H\sub{\s} &:= \frac{\omega}{2}(\prs{\varphi_+} - \prs{\varphi_-}), \nonumber \\
H\sub{\dd} &:= \lambda_+\prdd{\psi_+} + \lambda_-\prdd{\psi_-}. 
\end{align}
Furthermore, let the initial state of the system be 
\begin{equation}\label{eq:qubit-initial-state}
\rho\sub{\s} = q\prs{\varphi_+} + (1-q)\prs{\varphi_-},
\end{equation}
while that of $\dd$ is $\ket{\psi}$. We wish to measure a two-valued observable $M\sub{\s}$, with outcomes $\pm$, with the measurement model $\mathscr{M} = (\h\sub{\dd}, \ket{\psi}, U_M, Z\sub{\dd})$. In order to satisfy \reqref{requirement energy} for the measuring process, as shown by \lemref{lemma measurement energy conserving} and \lemref{Repeatability lemma}, $M\sub{\s}$ and $Z\sub{\dd}$ must commute with $H\sub{\s}$ and $H\sub{\dd}$, respectively. Therefore, we choose 
\begin{align}
M\sub{\s}:=\sum_{x\in \pm} x\prs{\varphi_x},
\end{align}
and 
\begin{align}
Z\sub{\dd}= \sum_{x\in \pm} x \prdd{\psi_x}.
\end{align}
Given our choice of $M\sub{\s}$ and $Z\sub{\dd}$,  the premeasurement unitary operator is chosen as
\begin{equation} 
U_{M}: \ket{\varphi_\pm} \otimes \ket{\psi} \mapsto \ket{\tilde\varphi_\pm}\otimes \ket{\psi_\pm}. 
\end{equation}
Finally, in order for the engine to satisfy  \reqref{requirement energy} and \reqref{requirement feedback} for the feedback process, we choose the  global feedback unitary operator
\begin{align}
V = \sum_{x\in \pm} U_x \otimes \prdd{\psi_\pm}. 
\end{align}

Following \cite{Aberg-Catalytic-Coherence}, we will use a harmonic oscillator of frequency $\omega$ as the weight, with the Hamiltonian
\begin{align}
H\sub{\ww} := \omega \sum_{n \in \nat} n \,\prw{n}.
\end{align}
Consequently, the conditional work extraction unitaries on $\ww+ \s$, namely, $U_\pm$, can be constructed as 
\begin{align}\label{conditional unitary}
U_\pm &:= \sum_{n=3}^\infty\sum_{a,b \in \{\pm\}} |n-f\>\<n-g| \otimes |\varphi_a\>\<\varphi_b| \<\varphi_a|G_\pm |\varphi_b\> \nonumber \\
& \, \, \, \, + \prw{1}\otimes\one\sub{\s},
\end{align}
where $f:= \max\{1,a1-b1\}$ and $g:= \max\{1,b1-a1\}$, with $a,b \in \{\pm\}$. Here, $G_\pm:= |\varphi_-\>\<\tilde \varphi_\pm| + |\varphi_+\>\<\tilde \varphi_\pm^\perp|$ is a unitary operator on $\s$, such that $\<\tilde \varphi_\pm| \tilde \varphi_\pm^\perp\>=0$. Therefore, when the system undergoes a transition $\ket{\varphi_+} \mapsto \ket{\varphi_-}$, the weight eigenstates are shifted up by one quantum, and vice versa.     

It can be easily verified that $[U_\pm, H\sub{\ww}+ H\sub{\s}]_-=\zero$, even when $\ket{\tilde \varphi_\pm}$ are not eigenstates of the system Hamiltonian. If the weight is initialised in a pure state $\rho\subw{\ww} := \prw{\Psi}$, where $\ket{\Psi}$ is an equal superposition of $N$ Hamiltonian eigenstates,
\begin{align}
\ket{\Psi}:= \frac{1}{\sqrt{N}} \sum_{n=2}^{N+1} \ket{n},
\end{align}
 then it can function as a work storage device. This is a result of the energy-translational invariance of $\ket{\Psi}$; adding or removing one quantum is identical to a coordinate transformation $n\mapsto n+1$ and $n\mapsto n-1$, respectively. Moreover, if $\ket{\tilde \varphi_\pm}$ are the eigenvectors of $H\sub{\s}$, then irrespective of $N$ the resulting dynamics on both $\s$ and $\ww$ will be unitary. As such, \featref{feature same entropy} will be satisfied in this case. This is not so when $\ket{\tilde \varphi_\pm}$ are superpositions of $H\sub{\s}$ eigenvectors. For example, in the case of $\ket{\tilde \varphi_\pm} = \frac{1}{\sqrt{2}}(\ket{\varphi_+} \pm \ket{\varphi_-})$, we have 
\begin{align}
\<\varphi_-|\Lambda_\pm(\tilde \varphi_\pm)|\varphi_-\> =\frac{2N-1}{2N},
\end{align}
with 
\begin{align}
S(\Lambda_\pm^*(\rho\subw{\ww})) < \frac{1}{2N}\ln\left(2N\right) + \frac{2N-1}{2N}\ln\left(\frac{2N}{2N-1}\right).
\end{align}
In the limit as $N$ tends to infinity, the increase in the weight's entropy can be made arbitrarily small, thus approximately satisfying \featref{feature same entropy}.

We now look at two possible implementations of measurement-assisted work extraction, labeled I and II. In I, the observable $M\sub{\s}$ is measured repeatably, thus satisfying \featref{feature repeatable measurement}, while in II this is not the case.  As the weight is initially pure, its entropy can never decrease. Therefore, \featref{feature positive work extraction} is satisfied in II, but not in I.

\subsection{Example I: repeatable measurement}

Let $\ket{\tilde\varphi_\pm} = \ket{\varphi_\pm}$, thus satisfying \featref{feature repeatable measurement}. Consequently, the state of $\s+\dd$ after premeasurement is
\begin{align}\label{eq:observable-commutes-qubit-measured-state}
U_M(\rho\sub{\s} \otimes \prdd{\psi})U_M^\dagger &= q P\sub{\ww+\s}[ \varphi_+\otimes \psi_+] \nonumber \\ &\, \, \, \, \, \,+ (1-q)P\sub{\ww+\s}[\varphi_-\otimes \psi_-].
\end{align}
Transforming this state with the weight by the global unitary $V$  prepares
\begin{align}\label{eq:observable-commutes-qubit-work-state}
&\tr\sub{\ww}[V U_M(\prw{\Psi}\otimes\rho\sub{\s} \otimes \prdd{\psi})U_M^\dagger V^\dagger ] \nonumber \\ &\, \, \, \,= q P\sub{\ww+\s}[\varphi_- \otimes \psi_+] +  (1-q)P\sub{\ww+\s}[\varphi_-\otimes \psi_-]  .
\end{align}
Comparing \eq{eq:observable-commutes-qubit-measured-state} with \eq{eq:observable-commutes-qubit-work-state}, we see that, as a result of feedback, the system undergoes the transition $\ket{\varphi_+}\mapsto \ket{\varphi_-}$ when the demon is in the state $\ket{\psi_+}$, resulting in a work extraction of $\omega$. When the demon is in the state $\ket{\psi_-}$, on the other hand, the system was already in the groundstate $\ket{\varphi_-}$ and is left the same, resulting in zero work extraction. Therefore, \featref{feature positive work extraction} is not satisfied.

\subsection{Example II: non-repeatable measurement}

Let $\ket{\tilde\varphi_\pm} = \frac{1}{\sqrt{2}}(\ket{\varphi_+} \pm \ket{\varphi_-})$. Hence, \featref{feature repeatable measurement} is not satisfied. Consequently, the state of $\s+\dd$ after premeasurement  is
\begin{align}\label{eq:observable-notcommutes-qubit-measured-state}
&U_{M}(\rho\sub{\s} \otimes \prdd{\psi})U_{M}^\dagger = q P\sub{\ww+\s}\left[\frac{(\varphi_+ + \varphi_-)}{\sqrt{2}}\otimes \psi_+ \right] \nonumber \\ & \, \, \, \, + (1-q)P\sub{\ww+\s}\left[\frac{(\varphi_+ - \varphi_-)}{\sqrt{2}}\otimes \psi_- \right].
\end{align}
Transforming this state with the weight by the global unitary $V$ prepares, in the ideal limit of $N \to \infty$, 
\begin{align}\label{eq:observable-notcommutes-qubit-work-state}
&\tr\sub{\ww}[ V U_{M}(\prw{\Psi}\otimes\rho\sub{\s} \otimes \prdd{\psi})U_{M}^\dagger V^\dagger] \nonumber \\
& \, \, \, = q P\sub{\ww+\s}[\varphi_- \otimes \psi_+] +  (1-q)P\sub{\ww+\s}[\varphi_-\otimes \psi_-] .
\end{align}
Comparing \eq{eq:observable-notcommutes-qubit-measured-state} with \eq{eq:observable-notcommutes-qubit-work-state} we see that, as a result of feedback, the system undergoes the transition $\frac{1}{\sqrt{2}}(\ket{\varphi_+}\pm \ket{\varphi_-})\mapsto \ket{\varphi_-}$ when the demon is in the states $\ket{\psi_\pm}$, resulting in a work extraction of $\omega/2$ for both measurement outcomes.  Therefore, \featref{feature positive work extraction} is  satisfied.

\section{Satisfying all three features with either a thermal reservoir, or degenerate observables}\label{circumventing theorem 1}

There are at least two ways in which \thmref{impossibility theorem} can be circumvented: (i)   letting the reservoir $\rr$ be involved during the feedback stage; and (ii) measure $\s$ with a degenerate observable.

\subsection{Szilard engine with heat from a thermal reservoir}
As a simple example, let $\s$ be a $d$-dimensional system, and let $\rr$ be a system initially prepared in the thermal state 
\begin{align}
\tau\sub{\rr}^\beta:= \frac{e^{-\beta H\sub{\rr}}}{\tr[e^{-\beta H\sub{\rr}}]},
\end{align} 
where $\beta = (K_B T)^{-1}$ is the inverse temperature. By  \lemref{Repeatability lemma},  the non-degenerate observable $M\sub{\s} = \sum_{x\in \xx} \pr{\varphi_x}$ can only be measured repeatably if it commutes with the system Hamiltonian $H\sub{\s}$. As such, in order to satisfy \featref{feature repeatable measurement} the post-measurement states $\{\ket{\varphi_x\}}_{x\in \xx}$ must be eigenstates of  $H\sub{\s}$. Including the reservoir in the feedback stage means that the $U_x$ in the feedback unitary operator defined in \eq{feedback unitary} are unitary operators on the compound $\ww+\s+\rr$ such that $[U_x,H\sub{\ww} + H\sub{\s} + H\sub{\rr}]_-=\zero$. The  CPTP maps defined in \eq{feedback CPTP maps} will therefore be modified as
\begin{align}\label{feedback CPTP maps reservoir}
\Lambda_x&: \prs{ \varphi_x}  \mapsto \tr\sub{\ww+\rr}[U_x(\rho\subw{\ww}\otimes\prs{ \varphi_x}\otimes \tau\sub{\rr}^\beta)U_x^\dagger], \nonumber \\
\Lambda_x '&: \tau\sub{\rr}^\beta  \mapsto \tr\sub{\ww+\s}[U_x(\rho\subw{\ww}\otimes\prs{ \varphi_x}\otimes \tau\sub{\rr}^\beta)U_x^\dagger], \nonumber \\
\Lambda_x^* &: \rho\subw{\ww} \mapsto \tr\sub{\s+\rr}[U_x(\rho\subw{\ww}\otimes\prs{ \varphi_x}\otimes \tau\sub{\rr}^\beta)U_x^\dagger].
\end{align}
 The subadditivity of the von-Neumann entropy and its invariance under unitary evolution implies that
\begin{align}\label{reservoir subadditivity entropy}
&S(\Lambda_x'[\tau\sub{\rr}^\beta])  - S(\tau\sub{\rr}^\beta) \nonumber \\ 
&\, \, \, \, \, \geqslant S(\rho\subw{\ww}) - S(\Lambda_x^*[\rho\subw{\ww}]) - S(\Lambda_x[\varphi_x]).
\end{align}

Recalling that when \featref{feature same entropy} is satisfied, $S(\rho\subw{\ww}) - S(\Lambda_x^*[\rho\subw{\ww}]) = 0$, then by \defref{definition work} and \eq{reservoir subadditivity entropy}, the  work that can be extracted for each measurement outcome, when both  \featref{feature repeatable measurement} and \featref{feature same entropy} are satisfied, is bounded by 
\begin{align}\label{reservoir-assisted-work-bound}
W_x &= \tr[H\sub{\rr}(\tau\sub{\rr}^\beta - \Lambda_x'[\tau\sub{\rr}^\beta])] + \tr[H\sub{\s}(\prs{\varphi_x} - \Lambda_x[\varphi_x])],  \nonumber \\
& =  \beta^{-1} \left( S(\tau\sub{\rr}^\beta) - S(\Lambda_x'[\tau\sub{\rr}^\beta])- S\left( \Lambda_x'[\tau\sub{\rr}^\beta]\| \tau\sub{\rr}^\beta\right)\right) \nonumber \\ 
&\, \, \, \, + \tr[H\sub{\s}(\prs{\varphi_x} - \Lambda_x[\varphi_x])] ,\nonumber \\
&\leqslant \beta^{-1} \left( S(\Lambda_x[\varphi_x])  - S( \Lambda_x'[\tau\sub{\rr}^\beta]\| \tau\sub{\rr}^\beta)\right)\nonumber \\ 
&\, \, \, \, + \tr[H\sub{\s}(\prs{\varphi_x} - \Lambda_x[\varphi_x])] ,\nonumber \\
&\leqslant \beta^{-1} S(\Lambda_x[\varphi_x])\nonumber  + \tr[H\sub{\s}(\prs{\varphi_x} - \Lambda_x[\varphi_x])] .
\end{align}
The final inequality can be saturated when the relative entropy term, $S( \Lambda_x'[\tau\sub{\rr}^\beta]\| \tau\sub{\rr}^\beta)$, which is a non-negative number, is made vanishingly small. As shown in \cite{Reeb-Wolf-Landauer}, this  can be done if the dimension of $\h\sub{\rr}$ is chosen to be sufficiently large, and its Hamiltonian spectrum is carefully chosen. As $S(\Lambda_x[\varphi_x])$ can be positive even when the weight's entropy is not allowed to change, we can always have positive work extraction. This is true even if the post-measurement state $\ket{\varphi_x}$ is the groundstate of $H\sub{\s}$. Moreover, if $H\sub{\s}$ is fully degenerate, and  $\Lambda_x[\varphi_x]= \one\sub{\s}/d$, then the maximum value of $W_x$ will be $K_B T \, \ln(d)$ for all $x\in \xx$. If $d=2$, this coincides with the work extracted from the classical Szilard engine when the volumes of the left and right side of the partition are identical.  

\subsection{Degenerate observables}\label{degenerate observable section}

Recall that \thmref{impossibility theorem} states that, when  \featref{feature same entropy} is satisfied, then the extracted work will not be positive for the outcome where the post-measurement state coincides with the groundstate of the system Hamiltonian. Here we show that, if the observable $M$ is both degenerate and is measured ``inefficiently'', then the post-measurement states can always be chosen so as to have more energy than the groundstate of $H\sub{\s}$, thus allowing for the circumvention of \thmref{impossibility theorem}.

For a system $\s$ with Hilbert space $\h\sub{\s} \simeq \co^d$ such that $d>2$, let $M\sub{\s}$ be a degenerate observable
\begin{align}
M\sub{\s} = \sum_{x\in \xx} x P^x\sub{\s},
\end{align}
such that $|\xx| <d$, and $\{P^x\sub{\s}\}_{x\in \xx}$ is a complete and orthogonal set of projection operators on $\h\sub{\s}$. We label the orthonormal eigenstates of $M\sub{\s}$ as $\ket{\varphi_x^\alpha}$, where $\alpha$ is a degeneracy label, such that $M\sub{\s} \ket{\varphi^\alpha_x} = x \ket{\varphi^\alpha_x}$  for all $\alpha$ and $x$.  The measurement model for this observable, $\mathscr{M} = (\h\sub{\dd}, \ket{\psi}, U_M, Z\sub{\dd})$,  will be repeatable if for all $x\in \xx$, the post-measurement states lie in the support of $P^x\sub{\s}$. Moreover, by the WAY theorem, if $\mathscr{M}$ is to be repeatable, given that $U_M$ conserves the total Hamiltonian, then $P\sub{\s}^x$ must commute with $H\sub{\s}$ for all $x\in \xx$.  Consider the projector $P\sub{\s}^y$ whose support contains the groundstate(s) of $H\sub{\s}$. It follows that for a repeatable measurement, $y$ is the only outcome whose post-measurement state will have support on the groundstate(s) of $H\sub{\s}$. Therefore, in order to circumvent \thmref{impossibility theorem} we need to show that, for all $\rho\sub{\s}$, the post-measurement state given outcome $y$ has more energy than the minimum eigenvalue of $H\sub{\s}$.

We will now look at two repeatable, and energy conserving  measurement models for the degenerate observable $M\sub{\s}$. The first model is a generalisation of a L\"uders measurement \cite{Mittelstaedt-measurement,Heinosaari}. Here, for some state $\rho\sub{\s}$, the post-measurement state of outcome $y$ is the groundstate of $H\sub{\s}$. Consequently, this measurement model will not circumvent  \thmref{impossibility theorem}. In the second model, we may always ensure that the post-measurement state for outcome $y$ will have more energy than the groundstate, thus circumventing \thmref{impossibility theorem}. We show that this is equivalent to coarse-graining the measurement outcomes of a non-degenerate observable, in such a way so as to allow for a repeatable measurement that is also ``inefficient''.

\subsubsection{Strong value-correlation measurements}

These measurements, just as the standard measurements for non-degenerate observables, have the property that, for any pure state $\ket{\Psi} \in \h\sub{\s}$, the post-measurement state for outcome $x\in \xx$ will also be pure. Here,  the premeasurement unitary operator is
\begin{align}
U_M: \ket{\varphi^\alpha_x}\otimes \ket{\psi} \mapsto  \ket{\tilde \varphi^\alpha_x}\otimes \ket{\psi_x},
\end{align}
where $\{ \ket{\tilde \varphi^\alpha_x}\}_\alpha$ is an orthonormal basis that spans $P\sub{\s}^x(\h\sub{\s})$. The instrument implemented by this measurement model will be
\begin{align}
\ii_x^M: \rho\sub{\s} \mapsto V_x P\sub{\s}^x \rho\sub{\s} P\sub{\s}^x V_x^\dagger,
\end{align}
where $V_x$ is a unitary operator acting on the support of $P\sub{\s}^x$. This instrument has only one Kraus operator, $K_x = V_x P\sub{\s}^x$, and it is said to result in an ``efficient'' measurement. If $V_x = \one$, whereby $\ket{\tilde \varphi^\alpha_x} = \ket{\varphi^\alpha_x}$, we have a L\"uders measurement. 

If the system is initially in the pure state 
\begin{align}
\ket{\Psi} = \sum_{x,\alpha} c_x^\alpha \ket{\varphi_x^\alpha},
\end{align} 
the post-measurement state for outcome $y$ will be
\begin{align}
&\frac{\ii_y^M(\prs{\Psi})}{\tr[\ii_y^M(\prs{\Psi})]} = \prs{ \Psi_y},\nonumber \\
&\ket{\Psi_y} = \frac{1}{N}\sum_{\alpha} c_y^\alpha \ket{\tilde \varphi^\alpha_y}, \, \, \, \, \, \, \, N^2 = \sum_\alpha |c^\alpha_y|^2.
\end{align}
Therefore, for some state $\ket{\Psi}$, the post-measurement state $\ket{\Psi_y}$ will be equal to the groundstate of the Hamiltonian. As such, \thmref{impossibility theorem} will not be circumvented.

\subsubsection{Coarse-grained standard measurements}
Let us denote the degenerate eigenstates of $Z\sub{\dd}$ as the orthonormal set of vectors  $\{\ket{\psi_x^\alpha}\}$  such that  $Z\sub{\dd} \ket{\psi^\alpha_x} = x \ket{\psi^\alpha_x}$ for all $x$ and $\alpha$.  The premeasurement unitary operator can then be defined as
\begin{align}\label{degenerate premeasurement}
U_M : \ket{\varphi^\alpha_x}\otimes \ket{\psi} \mapsto  \ket{\tilde \varphi^\alpha_x}\otimes \ket{\psi_x^\alpha}.
\end{align}
Comparing with \eq{standard premeasurement unitary}, we may see this as a coarse-grained measurement of a standard, non-degenerate observable. Now, the vectors in $\{\ket{\tilde \varphi^\alpha_x}\}_\alpha$ no longer have to be orthonormal. But, they must still be eigenstates of $M\sub{\s}$ with eigenvalue $x$ for the measurement to be repeatable. The instrument implemented by this measurement model will be
\begin{align}
\ii_x^M: \rho\sub{\s} \mapsto \sum_\alpha V_{x,\alpha} \prs{\varphi_x^\alpha} \rho\sub{\s} \prs{\varphi_x^\alpha} V_{x,\alpha}^\dagger,
\end{align}
where $V_{x,\alpha}$ are unitary operators acting on the support of $P\sub{\s}^x$. In contrast to the generalised L\"uders measurement discussed previously, this instrument has more than one Kraus operator, and leads to an ``inefficient'' measurement.  

If the system is initially in the pure state 
\begin{align}
\ket{\Psi} = \sum_{x,\alpha} c_x^\alpha \ket{\varphi_x^\alpha},
\end{align} 
the post-measurement state for outcome $y$ will be
\begin{align}
&\frac{\ii_y^M(\prs{\Psi})}{\tr[\ii_y^M(\prs{\Psi})]}  = \frac{1}{\sum_\alpha |c^\alpha_y|^2}\sum_\alpha |c^\alpha_y|^2 \prs{\tilde \varphi^\alpha_y}.
\end{align}
Due to the orthogonality of the vectors $\ket{\psi_x^\alpha}$ in \eq{degenerate premeasurement}, for each $x$ and $\alpha$, the vectors   $\ket{\tilde \varphi^\alpha_y}$ can  be any superpositions of Hamiltonian eigenstates that live in the support of $P^y\sub{\s}$. So we may simply choose these as the highest energy state within that subspace. Consequently, \thmref{impossibility theorem} will be circumvented.

\section{Net work extraction per cycle of a quantum Szilard engine without heat from a thermal reservoir}\label{Energetic contributions in one cycle of the quantum Szilard engine}

Each cycle of work extraction involves the following steps: (i) $\s$ is given in state $\rho\sub{\s}$; (ii)  $\s$ and $\dd$ undergo a joint unitary evolution by  $U_M$; (iii) work is extracted from $\s$ by a feedback unitary operator $V$ on $\ww+\s+\dd$; (iv) $\dd$ is reset to its initial state $\ket{\psi}$ by coupling to a thermal reservoir. \fig{work extraction model} shows this schematically.  

The initial state of the compound $\ww+\s+\dd+\rr$ is 
\begin{align}
\rho = \rho\sub{\ww}\otimes \rho\sub{\s} \otimes \prdd{\psi} \otimes \tau\sub{\rr}^\beta,
\end{align}
where $\tau\sub{\rr}^\beta := e^{-\beta H\sub{\rr}}/\tr[e^{-\beta H\sub{\rr}}]$ is the Gibbs state of the reservoir at inverse temperature $\beta = (K_B T)^{-1}$. After premeasurement, objectification, and feedback the state will be 
\begin{align}
\rho':= V (\rho\sub{\ww}\otimes\rho\sub{\s+\dd}^{M,O}) V^\dagger \otimes \tau\sub{\rr}^\beta,
\end{align}
where $\rho\sub{\s+\dd}^{M,O}$ is defined in \eq{Gemenge}.  The marginal states of $\rho'$ satisfy the relations 
\begin{align}
\rho\sub{\s}' &:= \sum_{x\in \xx} p^M_{\rho_{\!_\s}}(x)\Lambda_x[\tilde \varphi_x] \equiv \tr\sub{\ww+\dd}[V(\rho\subw{\ww}\otimes\rho^{M,O}\sub{\s+\dd}) V^\dagger] , \nonumber \\
\rho\subw{\ww}' &:= \sum_{x\in \xx} p^M_{\rho_{\!_\s}}(x)\Lambda_x^*[\rho\subw{\ww}] \equiv \tr\sub{\s+\dd}[V(\rho\subw{\ww}\otimes\rho^{M,O}\sub{\s+\dd}) V^\dagger], \nonumber \\
\rho\sub{\dd}' &:= \tr\sub{\ww+\s}[V(\rho\subw{\ww}\otimes\rho^{M,O}\sub{\s+\dd}) V^\dagger],
\end{align}
where $p^M_{\rho_{\!_\s}}(x)$ is the Born rule probability defined in \eq{Born rule}, while $\Lambda_x$ and $\Lambda_x^*$ are the CPTP maps induced by feedback, as defined in \eq{feedback CPTP maps}.

Using \defref{definition work}, we may view the work transferred into the weight, when the different measurement outcomes are not distinguished from one another, to be  
\begin{align} \label{coarse-grained work extraction} 
W_\xx &:= \tr[H\sub{\ww} (\rho\sub{\ww}' - \rho\sub{\ww} )] + K_B T \, (S(\rho\sub{\ww}) - S(\rho\sub{\ww}')), \nonumber \\
& = \tr[H\sub{\s} (\rho\sub{\s} - \rho\sub{\s}' )] + \tr[H\sub{\dd} (\prdd{\psi} - \rho\sub{\dd}' )]\nonumber \\ & \, \, \, + K_B T \, (S(\rho\sub{\ww}) - S(\rho\sub{\ww}')).
\end{align}
Here we have used the fact that feedback and measurement are energy conserving on the total system. We call $W_\xx$  the ``coarse-grained'' work, which is  different to the average work, obtained  by averaging $W_x$ over all measurement outcomes $x\in \xx$, which is 
\begin{align}\label{average  work extraction}
\<W_x\> &:= \sum_{x\in \xx} p^M_{\rho_{\!_\s}}(x) W_x, \nonumber \\ &= \tr[H\sub{\ww} (\rho\subw{\ww}' -\rho\subw{\ww})] \nonumber \\ & \, \, \, \, \, \, \, + K_B T \,\left(S(\rho\subw{\ww}) - \sum_{x\in \xx} p^M_{\rho_{\!_\s}}(x) S(\Lambda_x^*[\rho\subw{\ww}]) \right),\nonumber \\
& \geqslant W_\xx.
\end{align}
The inequality here is due to the concavity of the von-Neumann entropy.

Before the cycle can begin anew, the demon must be reset to the original pure state $\ket{\psi}$. This is achieved within the Landauer framework, by coupling $\dd$ with $\rr$ by the ``erasure'' unitary operator $U_R: \h\sub{\dd}\otimes \h\sub{\rr} \to \h\sub{\dd}\otimes \h\sub{\rr}$. If the reservoir is infinitely large, then $U_R$ can be chosen so that
\begin{align}
\tr\sub{\rr}[U_R (\rho\sub{\dd}' \otimes \tau\sub{\rr}^\beta)U_R^\dagger] = \prdd{\psi}.
\end{align}
To be sure, $U_R$ is generally  not energy conserving, and thus needs a hidden work source. Notwithstanding, this is not a problem,  because erasure always consumes work. Therefore, this hidden work source does not contribute to work extraction within a cycle. Defining the reduced state of the reservoir after its interaction with $\dd$ as $\tau\sub{\rr}'$, the consequent  increase in energy of the reservoir, defined as heat, obeys Landauer's inequality
\begin{equation}\label{Landauer inequality}
Q:= \tr[H\sub{\rr} (\tau\sub{\rr}' - \tau\sub{\rr}^\beta)] \geqslant \beta^{-1}  S(\rho\sub{\dd}').
\end{equation}
As shown in \cite{Reeb-Wolf-Landauer}, this bound can be achieved if the reservoir is infinitely large, and its Hamiltonian has a specific spectrum. Furthermore, we note that premeasurement, objectification, and feedback results in a unital CPTP map, which does not decrease the  von-Neumann entropy \cite{Uhlmann-Stochasticity, Nakahara-Decoherence}. This, together with the subadditivity of the von-Neumann entropy \cite{Petz-QI}, implies that
\begin{align}\label{feedback entropy inequality}
 S(\rho\subw{\ww}) + S(\rho\sub{\s}) &= S(\rho\subw{\ww}\otimes\rho\sub{\s}\otimes \prdd{\psi}) ,\nonumber \\ 
& \leqslant  S(V(\rho\subw{\ww}\otimes\rho\sub{\s+\dd}^{M,O}) V^\dagger), \nonumber \\
&\leqslant S(\rho'\subw{\ww})+ S(\rho\sub{\s}') + S(\rho\sub{\dd}').
\end{align}
Consequently, by combining \eq{Landauer inequality} and \eq{feedback entropy inequality}, and also taking into account the energy change of the demon due to erasure,  the work cost of erasure is shown to obey the inequality
\begin{align}\label{demon erasure energy cost}
W_R &:= \tr[H\sub{\dd}(\prdd{\psi} - \rho\sub{\dd}')] + Q, \nonumber \\ 
&\geqslant \tr[H\sub{\dd}(\prdd{\psi} - \rho\sub{\dd}')] \nonumber \\ 
& \, \, \, + K_B T \,  (S(\rho\subw{\ww}) + S(\rho\sub{\s}) - S(\rho'\subw{\ww}) - S(\rho\sub{\s}')).
\end{align}
Defining the net coarse-grained work extraction as $W^\mathrm{net}_\xx:= W_\xx - W_R$, by combining \eq{coarse-grained work extraction} and \eq{demon erasure energy cost} we arrive at the inequality
\begin{align}\label{obeying second law}
W^\mathrm{net}_\xx &= F(\rho\sub{\ww}') - F(\rho\sub{\ww}) - W_R\nonumber \\
&=  \tr[H\sub{\s} (\rho\sub{\s} - \rho\sub{\s}' )] \nonumber \\ & \, \, \, \, +  K_B T \, (S(\rho\sub{\ww}) - S(\rho\sub{\ww}')) - Q, \nonumber \\
& \leqslant F(\rho\sub{\s}) - F(\rho\sub{\s}').
\end{align}
The net average work extraction $\<W^\mathrm{net}_x\> := \<W_x\> - W_R$, on the other hand, obeys the modified inequality
\begin{align}\label{dissobeying second law}
\<W^\mathrm{net}_x\> &\leqslant F(\rho\sub{\s}) - F(\rho\sub{\s}') \nonumber \\
& + K_B T \,\left(S(\rho\subw{\ww}') - \sum_{x\in \xx} p^M_{\rho_{\!_\s}}(x) S(\Lambda_x^*[\rho\subw{\ww}]) \right).
\end{align}
Therefore, we see that while the coarse-grained work definition of \eq{coarse-grained work extraction} will satisfy the second law, the average work extraction defined in \eq{average work extraction} will not; if $\rho\sub{\s}$ is initially thermal, the net coarse-grained work extraction given by \eq{obeying second law} will never be positive, whereas the net average work extraction given by \eq{dissobeying second law} could be.


\begin{thebibliography}{59}%
\makeatletter
\providecommand \@ifxundefined [1]{%
 \@ifx{#1\undefined}
}%
\providecommand \@ifnum [1]{%
 \ifnum #1\expandafter \@firstoftwo
 \else \expandafter \@secondoftwo
 \fi
}%
\providecommand \@ifx [1]{%
 \ifx #1\expandafter \@firstoftwo
 \else \expandafter \@secondoftwo
 \fi
}%
\providecommand \natexlab [1]{#1}%
\providecommand \enquote  [1]{``#1''}%
\providecommand \bibnamefont  [1]{#1}%
\providecommand \bibfnamefont [1]{#1}%
\providecommand \citenamefont [1]{#1}%
\providecommand \href@noop [0]{\@secondoftwo}%
\providecommand \href [0]{\begingroup \@sanitize@url \@href}%
\providecommand \@href[1]{\@@startlink{#1}\@@href}%
\providecommand \@@href[1]{\endgroup#1\@@endlink}%
\providecommand \@sanitize@url [0]{\catcode `\\12\catcode `\$12\catcode
  `\&12\catcode `\#12\catcode `\^12\catcode `\_12\catcode `\%12\relax}%
\providecommand \@@startlink[1]{}%
\providecommand \@@endlink[0]{}%
\providecommand \url  [0]{\begingroup\@sanitize@url \@url }%
\providecommand \@url [1]{\endgroup\@href {#1}{\urlprefix }}%
\providecommand \urlprefix  [0]{URL }%
\providecommand \Eprint [0]{\href }%
\providecommand \doibase [0]{http://dx.doi.org/}%
\providecommand \selectlanguage [0]{\@gobble}%
\providecommand \bibinfo  [0]{\@secondoftwo}%
\providecommand \bibfield  [0]{\@secondoftwo}%
\providecommand \translation [1]{[#1]}%
\providecommand \BibitemOpen [0]{}%
\providecommand \bibitemStop [0]{}%
\providecommand \bibitemNoStop [0]{.\EOS\space}%
\providecommand \EOS [0]{\spacefactor3000\relax}%
\providecommand \BibitemShut  [1]{\csname bibitem#1\endcsname}%
\let\auto@bib@innerbib\@empty
\bibitem [{\citenamefont {Sagawa}\ and\ \citenamefont
  {Ueda}(2012)}]{Sagawa-Feedback-Control-Classical}%
  \BibitemOpen
  \bibfield  {author} {\bibinfo {author} {\bibfnamefont {T.}~\bibnamefont
  {Sagawa}}\ and\ \bibinfo {author} {\bibfnamefont {M.}~\bibnamefont {Ueda}},\
  }\href {\doibase 10.1103/PhysRevE.85.021104} {\bibfield  {journal} {\bibinfo
  {journal} {Phys. Rev. E}\ }\textbf {\bibinfo {volume} {85}},\ \bibinfo
  {pages} {021104} (\bibinfo {year} {2012})}\BibitemShut {NoStop}%
\bibitem [{\citenamefont {Shiraishi}\ \emph {et~al.}(2015)\citenamefont
  {Shiraishi}, \citenamefont {Ito}, \citenamefont {Kawaguchi},\ and\
  \citenamefont {Sagawa}}]{Sagawa-Measurement-Feedback-Autonomous}%
  \BibitemOpen
  \bibfield  {author} {\bibinfo {author} {\bibfnamefont {N.}~\bibnamefont
  {Shiraishi}}, \bibinfo {author} {\bibfnamefont {S.}~\bibnamefont {Ito}},
  \bibinfo {author} {\bibfnamefont {K.}~\bibnamefont {Kawaguchi}}, \ and\
  \bibinfo {author} {\bibfnamefont {T.}~\bibnamefont {Sagawa}},\ }\href
  {http://stacks.iop.org/1367-2630/17/i=4/a=045012} {\bibfield  {journal}
  {\bibinfo  {journal} {New J. Phys.}\ }\textbf {\bibinfo {volume} {17}},\
  \bibinfo {pages} {045012} (\bibinfo {year} {2015})}\BibitemShut {NoStop}%
\bibitem [{\citenamefont {{Maxwell}}(1871)}]{Maxwellian-demon}%
  \BibitemOpen
  \bibfield  {author} {\bibinfo {author} {\bibfnamefont {J.~C.}\ \bibnamefont
  {{Maxwell}}},\ }\href@noop {} {\emph {\bibinfo {title} {Theory of Heat}}}\
  (\bibinfo  {publisher} {Longmans},\ \bibinfo {year} {1871})\BibitemShut
  {NoStop}%
\bibitem [{\citenamefont {Maruyama}\ \emph {et~al.}(2009)\citenamefont
  {Maruyama}, \citenamefont {Nori},\ and\ \citenamefont
  {Vedral}}]{Maxwells-demon-colloquium}%
  \BibitemOpen
  \bibfield  {author} {\bibinfo {author} {\bibfnamefont {K.}~\bibnamefont
  {Maruyama}}, \bibinfo {author} {\bibfnamefont {F.}~\bibnamefont {Nori}}, \
  and\ \bibinfo {author} {\bibfnamefont {V.}~\bibnamefont {Vedral}},\ }\href
  {\doibase 10.1103/RevModPhys.81.1} {\bibfield  {journal} {\bibinfo  {journal}
  {Rev. Mod. Phys.}\ }\textbf {\bibinfo {volume} {81}},\ \bibinfo {pages} {1}
  (\bibinfo {year} {2009})}\BibitemShut {NoStop}%
\bibitem [{\citenamefont {Szilard}(1929)}]{Szilard}%
  \BibitemOpen
  \bibfield  {author} {\bibinfo {author} {\bibfnamefont {L.}~\bibnamefont
  {Szilard}},\ }\href@noop {} {\bibfield  {journal} {\bibinfo  {journal}
  {Zeitschrift fuer Physik}\ }\textbf {\bibinfo {volume} {53}},\ \bibinfo
  {pages} {840} (\bibinfo {year} {1929})}\BibitemShut {NoStop}%
\bibitem [{\citenamefont {Balian}(2007)}]{Balian-macro-1}%
  \BibitemOpen
  \bibfield  {author} {\bibinfo {author} {\bibfnamefont {R.}~\bibnamefont
  {Balian}},\ }\href@noop {} {\emph {\bibinfo {title} {{F}rom {M}icrophysics to
  {M}acrophysics: volume 1}}}\ (\bibinfo  {publisher} {Springer},\ \bibinfo
  {year} {2007})\BibitemShut {NoStop}%
\bibitem [{\citenamefont {{Penrose}}(1970)}]{Penrose-stat-mech}%
  \BibitemOpen
  \bibfield  {author} {\bibinfo {author} {\bibfnamefont {O.}~\bibnamefont
  {{Penrose}}},\ }\href@noop {} {\emph {\bibinfo {title} {Foundations of
  Statistical Mechanics: A Deductive Treatment}}}\ (\bibinfo  {publisher}
  {Pergamon},\ \bibinfo {year} {1970})\BibitemShut {NoStop}%
\bibitem [{\citenamefont {Bennett}(1982)}]{Bennett-Landauer-review}%
  \BibitemOpen
  \bibfield  {author} {\bibinfo {author} {\bibfnamefont {C.~H.}\ \bibnamefont
  {Bennett}},\ }\href@noop {} {\bibfield  {journal} {\bibinfo  {journal} {Int.
  J. Theor. Phys.}\ }\textbf {\bibinfo {volume} {21}},\ \bibinfo {pages} {905}
  (\bibinfo {year} {1982})}\BibitemShut {NoStop}%
\bibitem [{\citenamefont {Bennett}(2003)}]{Bennett-Landauer-Notes}%
  \BibitemOpen
  \bibfield  {author} {\bibinfo {author} {\bibfnamefont {C.~H.}\ \bibnamefont
  {Bennett}},\ }\href@noop {} {\bibfield  {journal} {\bibinfo  {journal}
  {Studies in History and Philosophy of Modern Physics}\ }\textbf {\bibinfo
  {volume} {34}},\ \bibinfo {pages} {501} (\bibinfo {year} {2003})}\BibitemShut
  {NoStop}%
\bibitem [{\citenamefont {Landauer}(1961)}]{Landauer}%
  \BibitemOpen
  \bibfield  {author} {\bibinfo {author} {\bibfnamefont {R.}~\bibnamefont
  {Landauer}},\ }\href@noop {} {\bibfield  {journal} {\bibinfo  {journal} {IBM
  J. Res. Dev.}\ }\textbf {\bibinfo {volume} {5}},\ \bibinfo {pages} {183}
  (\bibinfo {year} {1961})}\BibitemShut {NoStop}%
\bibitem [{\citenamefont {Landauer}(1996)}]{Landauer-information-physical}%
  \BibitemOpen
  \bibfield  {author} {\bibinfo {author} {\bibfnamefont {R.}~\bibnamefont
  {Landauer}},\ }\href {\doibase
  http://dx.doi.org/10.1016/0375-9601(96)00453-7} {\bibfield  {journal}
  {\bibinfo  {journal} {Physics Letters A}\ }\textbf {\bibinfo {volume}
  {217}},\ \bibinfo {pages} {188 } (\bibinfo {year} {1996})}\BibitemShut
  {NoStop}%
\bibitem [{\citenamefont {Reeb}\ and\ \citenamefont
  {Wolf}(2014)}]{Reeb-Wolf-Landauer}%
  \BibitemOpen
  \bibfield  {author} {\bibinfo {author} {\bibfnamefont {D.}~\bibnamefont
  {Reeb}}\ and\ \bibinfo {author} {\bibfnamefont {M.~M.}\ \bibnamefont
  {Wolf}},\ }\href {http://stacks.iop.org/1367-2630/16/i=10/a=103011}
  {\bibfield  {journal} {\bibinfo  {journal} {New J. Phys.}\ }\textbf {\bibinfo
  {volume} {16}},\ \bibinfo {pages} {103011} (\bibinfo {year}
  {2014})}\BibitemShut {NoStop}%
\bibitem [{\citenamefont {Vinjanampathy}\ and\ \citenamefont
  {Anders}(2016)}]{Anders-thermo-review}%
  \BibitemOpen
  \bibfield  {author} {\bibinfo {author} {\bibfnamefont {S.}~\bibnamefont
  {Vinjanampathy}}\ and\ \bibinfo {author} {\bibfnamefont {J.}~\bibnamefont
  {Anders}},\ }\href@noop {} {\bibfield  {journal} {\bibinfo  {journal}
  {Contemp. Phys.}\ }\textbf {\bibinfo {volume} {57}},\ \bibinfo {pages} {545}
  (\bibinfo {year} {2016})}\BibitemShut {NoStop}%
\bibitem [{\citenamefont {{Goold}}\ \emph {et~al.}(2016)\citenamefont
  {{Goold}}, \citenamefont {{Huber}}, \citenamefont {{Riera}}, \citenamefont
  {{del Rio}},\ and\ \citenamefont {{Skrzypczyk}}}]{Goold-thermo-review}%
  \BibitemOpen
  \bibfield  {author} {\bibinfo {author} {\bibfnamefont {J.}~\bibnamefont
  {{Goold}}}, \bibinfo {author} {\bibfnamefont {M.}~\bibnamefont {{Huber}}},
  \bibinfo {author} {\bibfnamefont {A.}~\bibnamefont {{Riera}}}, \bibinfo
  {author} {\bibfnamefont {L.}~\bibnamefont {{del Rio}}}, \ and\ \bibinfo
  {author} {\bibfnamefont {P.}~\bibnamefont {{Skrzypczyk}}},\ }\href
  {http://stacks.iop.org/1751-8121/49/i=14/a=143001} {\bibfield  {journal}
  {\bibinfo  {journal} {Journal of Physics A: Mathematical and Theoretical}\
  }\textbf {\bibinfo {volume} {49}},\ \bibinfo {pages} {143001} (\bibinfo
  {year} {2016})}\BibitemShut {NoStop}%
\bibitem [{\citenamefont {Millen}\ and\ \citenamefont
  {Xuereb}(2016)}]{Millen-thermo-review}%
  \BibitemOpen
  \bibfield  {author} {\bibinfo {author} {\bibfnamefont {J.}~\bibnamefont
  {Millen}}\ and\ \bibinfo {author} {\bibfnamefont {A.}~\bibnamefont
  {Xuereb}},\ }\href@noop {} {\bibfield  {journal} {\bibinfo  {journal} {New J.
  Phys.}\ }\textbf {\bibinfo {volume} {18}},\ \bibinfo {pages} {011002}
  (\bibinfo {year} {2016})}\BibitemShut {NoStop}%
\bibitem [{\citenamefont {Horodecki}\ and\ \citenamefont
  {Oppenheim}(2013)}]{Horodecki2013}%
  \BibitemOpen
  \bibfield  {author} {\bibinfo {author} {\bibfnamefont {M.}~\bibnamefont
  {Horodecki}}\ and\ \bibinfo {author} {\bibfnamefont {J.}~\bibnamefont
  {Oppenheim}},\ }\href {\doibase 10.1038/ncomms3059} {\bibfield  {journal}
  {\bibinfo  {journal} {Nature Communications}\ }\textbf {\bibinfo {volume}
  {4}},\ \bibinfo {pages} {2059} (\bibinfo {year} {2013})}\BibitemShut
  {NoStop}%
\bibitem [{\citenamefont {{Kammerlander}}\ and\ \citenamefont
  {{Anders}}(2015)}]{Anders-Measurement-Thermodynamics}%
  \BibitemOpen
  \bibfield  {author} {\bibinfo {author} {\bibfnamefont {P.}~\bibnamefont
  {{Kammerlander}}}\ and\ \bibinfo {author} {\bibfnamefont {J.}~\bibnamefont
  {{Anders}}},\ }\href@noop {} {\bibfield  {journal} {\bibinfo  {journal}
  {Scientific Reports}\ }\textbf {\bibinfo {volume} {6}},\ \bibinfo {pages}
  {22174} (\bibinfo {year} {2015})}\BibitemShut {NoStop}%
\bibitem [{\citenamefont {Lostaglio}\ \emph {et~al.}(2015)\citenamefont
  {Lostaglio}, \citenamefont {Jennings},\ and\ \citenamefont
  {Rudolph}}]{Lostaglio2015b}%
  \BibitemOpen
  \bibfield  {author} {\bibinfo {author} {\bibfnamefont {M.}~\bibnamefont
  {Lostaglio}}, \bibinfo {author} {\bibfnamefont {D.}~\bibnamefont {Jennings}},
  \ and\ \bibinfo {author} {\bibfnamefont {T.}~\bibnamefont {Rudolph}},\ }\href
  {\doibase 10.1038/ncomms7383} {\bibfield  {journal} {\bibinfo  {journal}
  {Nature Communications}\ }\textbf {\bibinfo {volume} {6}},\ \bibinfo {pages}
  {6383} (\bibinfo {year} {2015})}\BibitemShut {NoStop}%
\bibitem [{\citenamefont {Perarnau-Llobet}\ \emph {et~al.}(2015)\citenamefont
  {Perarnau-Llobet}, \citenamefont {Hovhannisyan}, \citenamefont {Huber},
  \citenamefont {Skrzypczyk}, \citenamefont {Brunner},\ and\ \citenamefont
  {Ac\'{\i}n}}]{Karen-extractable-work-correlations}%
  \BibitemOpen
  \bibfield  {author} {\bibinfo {author} {\bibfnamefont {M.}~\bibnamefont
  {Perarnau-Llobet}}, \bibinfo {author} {\bibfnamefont {K.~V.}\ \bibnamefont
  {Hovhannisyan}}, \bibinfo {author} {\bibfnamefont {M.}~\bibnamefont {Huber}},
  \bibinfo {author} {\bibfnamefont {P.}~\bibnamefont {Skrzypczyk}}, \bibinfo
  {author} {\bibfnamefont {N.}~\bibnamefont {Brunner}}, \ and\ \bibinfo
  {author} {\bibfnamefont {A.}~\bibnamefont {Ac\'{\i}n}},\ }\href {\doibase
  10.1103/PhysRevX.5.041011} {\bibfield  {journal} {\bibinfo  {journal} {Phys.
  Rev. X}\ }\textbf {\bibinfo {volume} {5}},\ \bibinfo {pages} {041011}
  (\bibinfo {year} {2015})}\BibitemShut {NoStop}%
\bibitem [{\citenamefont {Gogolin}\ and\ \citenamefont
  {Eisert}(2016)}]{Gogolin2015a}%
  \BibitemOpen
  \bibfield  {author} {\bibinfo {author} {\bibfnamefont {C.}~\bibnamefont
  {Gogolin}}\ and\ \bibinfo {author} {\bibfnamefont {J.}~\bibnamefont
  {Eisert}},\ }\href {http://arxiv.org/abs/1503.07538
  http://dx.doi.org/10.1088/0034-4885/79/5/056001
  http://stacks.iop.org/0034-4885/79/i=5/a=056001?key=crossref.c33d47bf51b1ef13ead58c10cab8b739}
  {\bibfield  {journal} {\bibinfo  {journal} {Reports on Progress in Physics}\
  }\textbf {\bibinfo {volume} {79}},\ \bibinfo {pages} {056001} (\bibinfo
  {year} {2016})}\BibitemShut {NoStop}%
\bibitem [{\citenamefont {Guryanova}\ \emph {et~al.}(2016)\citenamefont
  {Guryanova}, \citenamefont {Popescu}, \citenamefont {Short}, \citenamefont
  {Silva},\ and\ \citenamefont {Skrzypczyk}}]{Guryanova2015}%
  \BibitemOpen
  \bibfield  {author} {\bibinfo {author} {\bibfnamefont {Y.}~\bibnamefont
  {Guryanova}}, \bibinfo {author} {\bibfnamefont {S.}~\bibnamefont {Popescu}},
  \bibinfo {author} {\bibfnamefont {A.~J.}\ \bibnamefont {Short}}, \bibinfo
  {author} {\bibfnamefont {R.}~\bibnamefont {Silva}}, \ and\ \bibinfo {author}
  {\bibfnamefont {P.}~\bibnamefont {Skrzypczyk}},\ }\href {\doibase
  10.1038/ncomms12049} {\bibfield  {journal} {\bibinfo  {journal} {Nature
  Communications}\ }\textbf {\bibinfo {volume} {7}},\ \bibinfo {pages} {12049}
  (\bibinfo {year} {2016})}\BibitemShut {NoStop}%
\bibitem [{\citenamefont {{Yunger Halpern}}\ \emph {et~al.}(2016)\citenamefont
  {{Yunger Halpern}}, \citenamefont {Faist}, \citenamefont {Oppenheim},\ and\
  \citenamefont {Winter}}]{YungerHalpern2015a}%
  \BibitemOpen
  \bibfield  {author} {\bibinfo {author} {\bibfnamefont {N.}~\bibnamefont
  {{Yunger Halpern}}}, \bibinfo {author} {\bibfnamefont {P.}~\bibnamefont
  {Faist}}, \bibinfo {author} {\bibfnamefont {J.}~\bibnamefont {Oppenheim}}, \
  and\ \bibinfo {author} {\bibfnamefont {A.}~\bibnamefont {Winter}},\ }\href
  {\doibase 10.1038/ncomms12051} {\bibfield  {journal} {\bibinfo  {journal}
  {Nature Communications}\ }\textbf {\bibinfo {volume} {7}},\ \bibinfo {pages}
  {12051} (\bibinfo {year} {2016})}\BibitemShut {NoStop}%
\bibitem [{\citenamefont {Alhambra}\ \emph {et~al.}(2016)\citenamefont
  {Alhambra}, \citenamefont {Masanes}, \citenamefont {Oppenheim},\ and\
  \citenamefont {Perry}}]{Alhambra2016a}%
  \BibitemOpen
  \bibfield  {author} {\bibinfo {author} {\bibfnamefont {{\'{A}}.~M.}\
  \bibnamefont {Alhambra}}, \bibinfo {author} {\bibfnamefont {L.}~\bibnamefont
  {Masanes}}, \bibinfo {author} {\bibfnamefont {J.}~\bibnamefont {Oppenheim}},
  \ and\ \bibinfo {author} {\bibfnamefont {C.}~\bibnamefont {Perry}},\ }\href
  {\doibase 10.1103/PhysRevX.6.041017} {\bibfield  {journal} {\bibinfo
  {journal} {Phys. Rev. X}\ }\textbf {\bibinfo {volume} {6}},\ \bibinfo {pages}
  {041017} (\bibinfo {year} {2016})}\BibitemShut {NoStop}%
\bibitem [{\citenamefont {Zurek}(1986)}]{Zurek-Szilard}%
  \BibitemOpen
  \bibfield  {author} {\bibinfo {author} {\bibfnamefont {W.~H.}\ \bibnamefont
  {Zurek}},\ }\enquote {\bibinfo {title} {Maxwell's demon, szilard's engine and
  quantum measurements},}\ in\ \href {\doibase 10.1007/978-1-4613-2181-1_11}
  {\emph {\bibinfo {booktitle} {Frontiers of Nonequilibrium Statistical
  Physics}}},\ \bibinfo {editor} {edited by\ \bibinfo {editor} {\bibfnamefont
  {G.~T.}\ \bibnamefont {Moore}}\ and\ \bibinfo {editor} {\bibfnamefont
  {M.~O.}\ \bibnamefont {Scully}}}\ (\bibinfo  {publisher} {Springer US},\
  \bibinfo {address} {Boston, MA},\ \bibinfo {year} {1986})\ pp.\ \bibinfo
  {pages} {151--161}\BibitemShut {NoStop}%
\bibitem [{\citenamefont {{Plesch}}\ \emph {et~al.}(2014)\citenamefont
  {{Plesch}}, \citenamefont {{Dahlsten}}, \citenamefont {{Goold}},\ and\
  \citenamefont {{Vedral}}}]{Dahlsten-Szilard}%
  \BibitemOpen
  \bibfield  {author} {\bibinfo {author} {\bibfnamefont {M.}~\bibnamefont
  {{Plesch}}}, \bibinfo {author} {\bibfnamefont {O.}~\bibnamefont
  {{Dahlsten}}}, \bibinfo {author} {\bibfnamefont {J.}~\bibnamefont {{Goold}}},
  \ and\ \bibinfo {author} {\bibfnamefont {V.}~\bibnamefont {{Vedral}}},\
  }\href@noop {} {\bibfield  {journal} {\bibinfo  {journal} {Scientific
  Reports}\ }\textbf {\bibinfo {volume} {4}},\ \bibinfo {pages} {6995}
  (\bibinfo {year} {2014})}\BibitemShut {NoStop}%
\bibitem [{\citenamefont {Kim}\ \emph {et~al.}(2011)\citenamefont {Kim},
  \citenamefont {Sagawa}, \citenamefont {De~Liberato},\ and\ \citenamefont
  {Ueda}}]{Sagawa-Szilard}%
  \BibitemOpen
  \bibfield  {author} {\bibinfo {author} {\bibfnamefont {S.~W.}\ \bibnamefont
  {Kim}}, \bibinfo {author} {\bibfnamefont {T.}~\bibnamefont {Sagawa}},
  \bibinfo {author} {\bibfnamefont {S.}~\bibnamefont {De~Liberato}}, \ and\
  \bibinfo {author} {\bibfnamefont {M.}~\bibnamefont {Ueda}},\ }\href {\doibase
  10.1103/PhysRevLett.106.070401} {\bibfield  {journal} {\bibinfo  {journal}
  {Phys. Rev. Lett.}\ }\textbf {\bibinfo {volume} {106}},\ \bibinfo {pages}
  {070401} (\bibinfo {year} {2011})}\BibitemShut {NoStop}%
\bibitem [{\citenamefont {Sagawa}\ and\ \citenamefont
  {Ueda}(2008)}]{Sagawa-feedback-control-Quantum}%
  \BibitemOpen
  \bibfield  {author} {\bibinfo {author} {\bibfnamefont {T.}~\bibnamefont
  {Sagawa}}\ and\ \bibinfo {author} {\bibfnamefont {M.}~\bibnamefont {Ueda}},\
  }\href {\doibase 10.1103/PhysRevLett.100.080403} {\bibfield  {journal}
  {\bibinfo  {journal} {Phys. Rev. Lett.}\ }\textbf {\bibinfo {volume} {100}},\
  \bibinfo {pages} {080403} (\bibinfo {year} {2008})}\BibitemShut {NoStop}%
\bibitem [{\citenamefont {Jacobs}(2009)}]{Jacobs-Feedback-freeenergy}%
  \BibitemOpen
  \bibfield  {author} {\bibinfo {author} {\bibfnamefont {K.}~\bibnamefont
  {Jacobs}},\ }\href {\doibase 10.1103/PhysRevA.80.012322} {\bibfield
  {journal} {\bibinfo  {journal} {Phys. Rev. A}\ }\textbf {\bibinfo {volume}
  {80}},\ \bibinfo {pages} {012322} (\bibinfo {year} {2009})}\BibitemShut
  {NoStop}%
\bibitem [{\citenamefont {Park}\ \emph {et~al.}(2013)\citenamefont {Park},
  \citenamefont {Kim}, \citenamefont {Sagawa},\ and\ \citenamefont
  {Kim}}]{Sagawa-Heat-engine-QI}%
  \BibitemOpen
  \bibfield  {author} {\bibinfo {author} {\bibfnamefont {J.~J.}\ \bibnamefont
  {Park}}, \bibinfo {author} {\bibfnamefont {K.-H.}\ \bibnamefont {Kim}},
  \bibinfo {author} {\bibfnamefont {T.}~\bibnamefont {Sagawa}}, \ and\ \bibinfo
  {author} {\bibfnamefont {S.~W.}\ \bibnamefont {Kim}},\ }\href {\doibase
  10.1103/PhysRevLett.111.230402} {\bibfield  {journal} {\bibinfo  {journal}
  {Phys. Rev. Lett.}\ }\textbf {\bibinfo {volume} {111}},\ \bibinfo {pages}
  {230402} (\bibinfo {year} {2013})}\BibitemShut {NoStop}%
\bibitem [{\citenamefont {Camati}\ \emph {et~al.}(2016)\citenamefont {Camati},
  \citenamefont {Peterson}, \citenamefont {Batalh\~ao}, \citenamefont
  {Micadei}, \citenamefont {Souza}, \citenamefont {Sarthour}, \citenamefont
  {Oliveira},\ and\ \citenamefont {Serra}}]{Camati-Maxwell-Demon}%
  \BibitemOpen
  \bibfield  {author} {\bibinfo {author} {\bibfnamefont {P.~A.}\ \bibnamefont
  {Camati}}, \bibinfo {author} {\bibfnamefont {J.~P.~S.}\ \bibnamefont
  {Peterson}}, \bibinfo {author} {\bibfnamefont {T.~B.}\ \bibnamefont
  {Batalh\~ao}}, \bibinfo {author} {\bibfnamefont {K.}~\bibnamefont {Micadei}},
  \bibinfo {author} {\bibfnamefont {A.~M.}\ \bibnamefont {Souza}}, \bibinfo
  {author} {\bibfnamefont {R.~S.}\ \bibnamefont {Sarthour}}, \bibinfo {author}
  {\bibfnamefont {I.~S.}\ \bibnamefont {Oliveira}}, \ and\ \bibinfo {author}
  {\bibfnamefont {R.~M.}\ \bibnamefont {Serra}},\ }\href {\doibase
  10.1103/PhysRevLett.117.240502} {\bibfield  {journal} {\bibinfo  {journal}
  {Phys. Rev. Lett.}\ }\textbf {\bibinfo {volume} {117}},\ \bibinfo {pages}
  {240502} (\bibinfo {year} {2016})}\BibitemShut {NoStop}%
\bibitem [{\citenamefont {{Cottet}}\ \emph {et~al.}()\citenamefont {{Cottet}},
  \citenamefont {{Jezouin}}, \citenamefont {{Bretheau}}, \citenamefont
  {{Campagne-Ibarcq}}, \citenamefont {{Ficheux}}, \citenamefont {{Anders}},
  \citenamefont {{Auff{\`e}ves}}, \citenamefont {{Azouit}}, \citenamefont
  {{Rouchon}},\ and\ \citenamefont {{Huard}}}]{Janet-Maxwell-Experiment}%
  \BibitemOpen
  \bibfield  {author} {\bibinfo {author} {\bibfnamefont {N.}~\bibnamefont
  {{Cottet}}}, \bibinfo {author} {\bibfnamefont {S.}~\bibnamefont {{Jezouin}}},
  \bibinfo {author} {\bibfnamefont {L.}~\bibnamefont {{Bretheau}}}, \bibinfo
  {author} {\bibfnamefont {P.}~\bibnamefont {{Campagne-Ibarcq}}}, \bibinfo
  {author} {\bibfnamefont {Q.}~\bibnamefont {{Ficheux}}}, \bibinfo {author}
  {\bibfnamefont {J.}~\bibnamefont {{Anders}}}, \bibinfo {author}
  {\bibfnamefont {A.}~\bibnamefont {{Auff{\`e}ves}}}, \bibinfo {author}
  {\bibfnamefont {R.}~\bibnamefont {{Azouit}}}, \bibinfo {author}
  {\bibfnamefont {P.}~\bibnamefont {{Rouchon}}}, \ and\ \bibinfo {author}
  {\bibfnamefont {B.}~\bibnamefont {{Huard}}},\ }\href@noop {} {\bibinfo
  {journal} {ArXiv: 1702.05161}\ }\BibitemShut {NoStop}%
\bibitem [{\citenamefont {{Elouard}}\ \emph {et~al.}(2017)\citenamefont
  {{Elouard}}, \citenamefont {{Herrera-Mart{\'{\i}}}}, \citenamefont
  {{Clusel}},\ and\ \citenamefont
  {{Auff{\`e}ves}}}]{Alexia-thermo-Measurement}%
  \BibitemOpen
\bibfield  {journal} {  }\bibfield  {author} {\bibinfo {author} {\bibfnamefont
  {C.}~\bibnamefont {{Elouard}}}, \bibinfo {author} {\bibfnamefont
  {D.}~\bibnamefont {{Herrera-Mart{\'{\i}}}}}, \bibinfo {author} {\bibfnamefont
  {M.}~\bibnamefont {{Clusel}}}, \ and\ \bibinfo {author} {\bibfnamefont
  {A.}~\bibnamefont {{Auff{\`e}ves}}},\ }\href@noop {} {\bibfield  {journal}
  {\bibinfo  {journal} {npj Quantum Information}\ }\textbf {\bibinfo {volume}
  {3}} (\bibinfo {year} {2017})}\BibitemShut {NoStop}%
\bibitem [{\citenamefont {Elouard}\ \emph {et~al.}(2017)\citenamefont
  {Elouard}, \citenamefont {Herrera-Mart{\'{i}}}, \citenamefont {Huard},\ and\
  \citenamefont {Auff{\`{e}}ves}}]{Alexia-Maxwell-Measurement}%
  \BibitemOpen
  \bibfield  {author} {\bibinfo {author} {\bibfnamefont {C.}~\bibnamefont
  {Elouard}}, \bibinfo {author} {\bibfnamefont {D.}~\bibnamefont
  {Herrera-Mart{\'{i}}}}, \bibinfo {author} {\bibfnamefont {B.}~\bibnamefont
  {Huard}}, \ and\ \bibinfo {author} {\bibfnamefont {A.}~\bibnamefont
  {Auff{\`{e}}ves}},\ }\href {\doibase 10.1103/PhysRevLett.118.260603}
  {\bibfield  {journal} {\bibinfo  {journal} {Phys. Rev. Lett.}\ }\textbf
  {\bibinfo {volume} {118}},\ \bibinfo {pages} {260603} (\bibinfo {year}
  {2017})}\BibitemShut {NoStop}%
\bibitem [{\citenamefont {Sagawa}\ and\ \citenamefont
  {Ueda}(2009)}]{Sagawa-thermodynamic-measurement}%
  \BibitemOpen
  \bibfield  {author} {\bibinfo {author} {\bibfnamefont {T.}~\bibnamefont
  {Sagawa}}\ and\ \bibinfo {author} {\bibfnamefont {M.}~\bibnamefont {Ueda}},\
  }\href {\doibase 10.1103/PhysRevLett.102.250602} {\bibfield  {journal}
  {\bibinfo  {journal} {Phys. Rev. Lett.}\ }\textbf {\bibinfo {volume} {102}},\
  \bibinfo {pages} {250602} (\bibinfo {year} {2009})}\BibitemShut {NoStop}%
\bibitem [{\citenamefont {Jacobs}(2012)}]{Kurt-measurement-thermo}%
  \BibitemOpen
  \bibfield  {author} {\bibinfo {author} {\bibfnamefont {K.}~\bibnamefont
  {Jacobs}},\ }\href {\doibase 10.1103/PhysRevE.86.040106} {\bibfield
  {journal} {\bibinfo  {journal} {Phys. Rev. E}\ }\textbf {\bibinfo {volume}
  {86}},\ \bibinfo {pages} {040106} (\bibinfo {year} {2012})}\BibitemShut
  {NoStop}%
\bibitem [{\citenamefont {Navascu\'es}\ and\ \citenamefont
  {Popescu}(2014)}]{Popescu-energy-conversation-measurement}%
  \BibitemOpen
  \bibfield  {author} {\bibinfo {author} {\bibfnamefont {M.}~\bibnamefont
  {Navascu\'es}}\ and\ \bibinfo {author} {\bibfnamefont {S.}~\bibnamefont
  {Popescu}},\ }\href {\doibase 10.1103/PhysRevLett.112.140502} {\bibfield
  {journal} {\bibinfo  {journal} {Phys. Rev. Lett.}\ }\textbf {\bibinfo
  {volume} {112}},\ \bibinfo {pages} {140502} (\bibinfo {year}
  {2014})}\BibitemShut {NoStop}%
\bibitem [{\citenamefont {Miyadera}(2016)}]{Miyadera-Time-Energy-Measurement}%
  \BibitemOpen
  \bibfield  {author} {\bibinfo {author} {\bibfnamefont {T.}~\bibnamefont
  {Miyadera}},\ }\href {\doibase 10.1007/s10701-016-0027-6} {\bibfield
  {journal} {\bibinfo  {journal} {Foundations of Physics}\ }\textbf {\bibinfo
  {volume} {46}},\ \bibinfo {pages} {1522} (\bibinfo {year}
  {2016})}\BibitemShut {NoStop}%
\bibitem [{\citenamefont {{Abdelkhalek}}\ \emph {et~al.}()\citenamefont
  {{Abdelkhalek}}, \citenamefont {{Nakata}},\ and\ \citenamefont
  {{Reeb}}}]{Abdelkhalek-measurement}%
  \BibitemOpen
  \bibfield  {author} {\bibinfo {author} {\bibfnamefont {K.}~\bibnamefont
  {{Abdelkhalek}}}, \bibinfo {author} {\bibfnamefont {Y.}~\bibnamefont
  {{Nakata}}}, \ and\ \bibinfo {author} {\bibfnamefont {D.}~\bibnamefont
  {{Reeb}}},\ }\href@noop {} {\bibinfo  {journal} {ArXiv: 1609.06981}\
  }\BibitemShut {NoStop}%
\bibitem [{\citenamefont {{Wigner}}(1952)}]{Wigner-Measurement-conservation}%
  \BibitemOpen
\bibfield  {journal} {  }\bibfield  {author} {\bibinfo {author} {\bibfnamefont
  {E.}~\bibnamefont {{Wigner}}},\ }\href@noop {} {\bibfield  {journal}
  {\bibinfo  {journal} {Z. Phys.}\ }\textbf {\bibinfo {volume} {133}},\
  \bibinfo {pages} {101} (\bibinfo {year} {1952})}\BibitemShut {NoStop}%
\bibitem [{\citenamefont {Araki}\ and\ \citenamefont
  {Yanase}(1960)}]{Araki-Yanase}%
  \BibitemOpen
  \bibfield  {author} {\bibinfo {author} {\bibfnamefont {H.}~\bibnamefont
  {Araki}}\ and\ \bibinfo {author} {\bibfnamefont {M.~M.}\ \bibnamefont
  {Yanase}},\ }\href {\doibase 10.1103/PhysRev.120.622} {\bibfield  {journal}
  {\bibinfo  {journal} {Phys. Rev.}\ }\textbf {\bibinfo {volume} {120}},\
  \bibinfo {pages} {622} (\bibinfo {year} {1960})}\BibitemShut {NoStop}%
\bibitem [{\citenamefont {Miyadera}\ and\ \citenamefont
  {Imai}(2006)}]{Miyadera-WAY-distinguishability}%
  \BibitemOpen
  \bibfield  {author} {\bibinfo {author} {\bibfnamefont {T.}~\bibnamefont
  {Miyadera}}\ and\ \bibinfo {author} {\bibfnamefont {H.}~\bibnamefont
  {Imai}},\ }\href {\doibase 10.1103/PhysRevA.74.024101} {\bibfield  {journal}
  {\bibinfo  {journal} {Phys. Rev. A}\ }\textbf {\bibinfo {volume} {74}},\
  \bibinfo {pages} {024101} (\bibinfo {year} {2006})}\BibitemShut {NoStop}%
\bibitem [{\citenamefont {Loveridge}\ and\ \citenamefont
  {Busch}(2011)}]{Loveridge-WAY}%
  \BibitemOpen
  \bibfield  {author} {\bibinfo {author} {\bibfnamefont {L.}~\bibnamefont
  {Loveridge}}\ and\ \bibinfo {author} {\bibfnamefont {P.}~\bibnamefont
  {Busch}},\ }\href {\doibase 10.1140/epjd/e2011-10714-3} {\bibfield  {journal}
  {\bibinfo  {journal} {The European Physical Journal D}\ }\textbf {\bibinfo
  {volume} {62}},\ \bibinfo {pages} {297} (\bibinfo {year} {2011})}\BibitemShut
  {NoStop}%
\bibitem [{\citenamefont {Busch}\ and\ \citenamefont
  {Loveridge}(2011)}]{Loveridge-WAY-2}%
  \BibitemOpen
  \bibfield  {author} {\bibinfo {author} {\bibfnamefont {P.}~\bibnamefont
  {Busch}}\ and\ \bibinfo {author} {\bibfnamefont {L.}~\bibnamefont
  {Loveridge}},\ }\href {\doibase 10.1103/PhysRevLett.106.110406} {\bibfield
  {journal} {\bibinfo  {journal} {Phys. Rev. Lett.}\ }\textbf {\bibinfo
  {volume} {106}},\ \bibinfo {pages} {110406} (\bibinfo {year}
  {2011})}\BibitemShut {NoStop}%
\bibitem [{\citenamefont {Ahmadi}\ \emph {et~al.}(2013)\citenamefont {Ahmadi},
  \citenamefont {Jennings},\ and\ \citenamefont {Rudolph}}]{Mehdi-WAY}%
  \BibitemOpen
  \bibfield  {author} {\bibinfo {author} {\bibfnamefont {M.}~\bibnamefont
  {Ahmadi}}, \bibinfo {author} {\bibfnamefont {D.}~\bibnamefont {Jennings}}, \
  and\ \bibinfo {author} {\bibfnamefont {T.}~\bibnamefont {Rudolph}},\ }\href
  {http://stacks.iop.org/1367-2630/15/i=1/a=013057} {\bibfield  {journal}
  {\bibinfo  {journal} {New J. Phys.}\ }\textbf {\bibinfo {volume} {15}},\
  \bibinfo {pages} {013057} (\bibinfo {year} {2013})}\BibitemShut {NoStop}%
\bibitem [{\citenamefont {Skrzypczyk}\ \emph {et~al.}(2014)\citenamefont
  {Skrzypczyk}, \citenamefont {Short},\ and\ \citenamefont
  {Popescu}}]{thermo-individual-quantum}%
  \BibitemOpen
  \bibfield  {author} {\bibinfo {author} {\bibfnamefont {P.}~\bibnamefont
  {Skrzypczyk}}, \bibinfo {author} {\bibfnamefont {A.~J.}\ \bibnamefont
  {Short}}, \ and\ \bibinfo {author} {\bibfnamefont {S.}~\bibnamefont
  {Popescu}},\ }\href@noop {} {\bibfield  {journal} {\bibinfo  {journal}
  {Nature Communications}\ }\textbf {\bibinfo {volume} {5}} (\bibinfo {year}
  {2014})}\BibitemShut {NoStop}%
\bibitem [{\citenamefont {\AA{}berg}(2014)}]{Aberg-Catalytic-Coherence}%
  \BibitemOpen
  \bibfield  {author} {\bibinfo {author} {\bibfnamefont {J.}~\bibnamefont
  {\AA{}berg}},\ }\href {\doibase 10.1103/PhysRevLett.113.150402} {\bibfield
  {journal} {\bibinfo  {journal} {Phys. Rev. Lett.}\ }\textbf {\bibinfo
  {volume} {113}},\ \bibinfo {pages} {150402} (\bibinfo {year}
  {2014})}\BibitemShut {NoStop}%
\bibitem [{\citenamefont {von {Neumann}}(1996)}]{von-Neumann}%
  \BibitemOpen
  \bibfield  {author} {\bibinfo {author} {\bibfnamefont {J.}~\bibnamefont {von
  {Neumann}}},\ }\href@noop {} {\emph {\bibinfo {title} {{M}athematical
  {F}oundations of {Q}uantum {M}echanics}}}\ (\bibinfo  {publisher} {Princeton
  University Press},\ \bibinfo {year} {1996})\BibitemShut {NoStop}%
\bibitem [{\citenamefont {{Busch}}\ \emph {et~al.}(1995)\citenamefont
  {{Busch}}, \citenamefont {{Grabowski}},\ and\ \citenamefont
  {{Lahti}}}]{Busch-operational}%
  \BibitemOpen
  \bibfield  {author} {\bibinfo {author} {\bibfnamefont {P.}~\bibnamefont
  {{Busch}}}, \bibinfo {author} {\bibfnamefont {M.}~\bibnamefont
  {{Grabowski}}}, \ and\ \bibinfo {author} {\bibfnamefont {P.~J.}\ \bibnamefont
  {{Lahti}}},\ }\href@noop {} {\emph {\bibinfo {title} {{O}perational {Q}uantum
  {P}hysics}}}\ (\bibinfo  {publisher} {Springer},\ \bibinfo {year}
  {1995})\BibitemShut {NoStop}%
\bibitem [{\citenamefont {{Busch}}\ \emph {et~al.}(1996)\citenamefont
  {{Busch}}, \citenamefont {{Lahti}},\ and\ \citenamefont
  {{Mittelstaedt}}}]{Busch-measurement}%
  \BibitemOpen
  \bibfield  {author} {\bibinfo {author} {\bibfnamefont {P.}~\bibnamefont
  {{Busch}}}, \bibinfo {author} {\bibfnamefont {P.~J.}\ \bibnamefont
  {{Lahti}}}, \ and\ \bibinfo {author} {\bibfnamefont {P.}~\bibnamefont
  {{Mittelstaedt}}},\ }\href@noop {} {\emph {\bibinfo {title} {{T}he {Q}uantum
  {T}heory of {M}easurement}}}\ (\bibinfo  {publisher} {Springer},\ \bibinfo
  {year} {1996})\BibitemShut {NoStop}%
\bibitem [{\citenamefont {{Busch}}\ \emph {et~al.}(2016)\citenamefont
  {{Busch}}, \citenamefont {{Lahti}}, \citenamefont {{Pellonp\"a\"a}},\ and\
  \citenamefont {{Ylinen}}}]{Busch-measurement-2}%
  \BibitemOpen
  \bibfield  {author} {\bibinfo {author} {\bibfnamefont {P.}~\bibnamefont
  {{Busch}}}, \bibinfo {author} {\bibfnamefont {P.~J.}\ \bibnamefont
  {{Lahti}}}, \bibinfo {author} {\bibfnamefont {J.~P.}\ \bibnamefont
  {{Pellonp\"a\"a}}}, \ and\ \bibinfo {author} {\bibfnamefont {K.}~\bibnamefont
  {{Ylinen}}},\ }\href@noop {} {\emph {\bibinfo {title} {{Q}uantum
  {M}easurement}}}\ (\bibinfo  {publisher} {Springer},\ \bibinfo {year}
  {2016})\BibitemShut {NoStop}%
\bibitem [{\citenamefont {Heinosaari}\ and\ \citenamefont
  {Ziman}(2011)}]{Heinosaari}%
  \BibitemOpen
  \bibfield  {author} {\bibinfo {author} {\bibfnamefont {T.}~\bibnamefont
  {Heinosaari}}\ and\ \bibinfo {author} {\bibfnamefont {M.}~\bibnamefont
  {Ziman}},\ }\href@noop {} {\emph {\bibinfo {title} {{T}he {M}athematical
  {L}anguage of {Q}uantum {T}heory}}}\ (\bibinfo  {publisher} {Cambridge
  University Press},\ \bibinfo {year} {2011})\BibitemShut {NoStop}%
\bibitem [{\citenamefont {{Mittelstaedt}}(2004)}]{Mittelstaedt-measurement}%
  \BibitemOpen
  \bibfield  {author} {\bibinfo {author} {\bibfnamefont {P.}~\bibnamefont
  {{Mittelstaedt}}},\ }\href@noop {} {\emph {\bibinfo {title} {{T}he
  {I}nterpretation of {Q}uantum {M}echanics and the {M}easurement {P}rocess}}}\
  (\bibinfo  {publisher} {Cambridge University Press},\ \bibinfo {year}
  {2004})\BibitemShut {NoStop}%
\bibitem [{\citenamefont {Gemmer}\ and\ \citenamefont
  {Anders}(2015)}]{Gemmer2015}%
  \BibitemOpen
  \bibfield  {author} {\bibinfo {author} {\bibfnamefont {J.}~\bibnamefont
  {Gemmer}}\ and\ \bibinfo {author} {\bibfnamefont {J.}~\bibnamefont
  {Anders}},\ }\href {\doibase 10.1088/1367-2630/17/8/085006} {\bibfield
  {journal} {\bibinfo  {journal} {New J. Phys.}\ }\textbf {\bibinfo {volume}
  {17}},\ \bibinfo {pages} {085006} (\bibinfo {year} {2015})}\BibitemShut
  {NoStop}%
\bibitem [{\citenamefont {Gallego}\ \emph {et~al.}(2016)\citenamefont
  {Gallego}, \citenamefont {Eisert},\ and\ \citenamefont
  {Wilming}}]{Gallego2015-Work-definition}%
  \BibitemOpen
  \bibfield  {author} {\bibinfo {author} {\bibfnamefont {R.}~\bibnamefont
  {Gallego}}, \bibinfo {author} {\bibfnamefont {J.}~\bibnamefont {Eisert}}, \
  and\ \bibinfo {author} {\bibfnamefont {H.}~\bibnamefont {Wilming}},\ }\href
  {\doibase 10.1088/1367-2630/18/10/103017} {\bibfield  {journal} {\bibinfo
  {journal} {New J. Phys.}\ }\textbf {\bibinfo {volume} {18}},\ \bibinfo
  {pages} {103017} (\bibinfo {year} {2016})}\BibitemShut {NoStop}%
\bibitem [{\citenamefont {Morikuni}\ \emph {et~al.}(2017)\citenamefont
  {Morikuni}, \citenamefont {Tajima},\ and\ \citenamefont
  {Hatano}}]{Morikuni2017}%
  \BibitemOpen
  \bibfield  {author} {\bibinfo {author} {\bibfnamefont {Y.}~\bibnamefont
  {Morikuni}}, \bibinfo {author} {\bibfnamefont {H.}~\bibnamefont {Tajima}}, \
  and\ \bibinfo {author} {\bibfnamefont {N.}~\bibnamefont {Hatano}},\ }\href
  {\doibase 10.1103/PhysRevE.95.032147} {\bibfield  {journal} {\bibinfo
  {journal} {Physi. Rev. E}\ }\textbf {\bibinfo {volume} {95}},\ \bibinfo
  {pages} {032147} (\bibinfo {year} {2017})}\BibitemShut {NoStop}%
\bibitem [{\citenamefont {Anders}\ and\ \citenamefont
  {Giovannetti}(2013)}]{Anders-discrete-thermo}%
  \BibitemOpen
  \bibfield  {author} {\bibinfo {author} {\bibfnamefont {J.}~\bibnamefont
  {Anders}}\ and\ \bibinfo {author} {\bibfnamefont {V.}~\bibnamefont
  {Giovannetti}},\ }\href {http://stacks.iop.org/1367-2630/15/i=3/a=033022}
  {\bibfield  {journal} {\bibinfo  {journal} {New J. Phys.}\ }\textbf {\bibinfo
  {volume} {15}},\ \bibinfo {pages} {033022} (\bibinfo {year}
  {2013})}\BibitemShut {NoStop}%
\bibitem [{\citenamefont {Alberti}\ and\ \citenamefont
  {Uhlmann}(1982)}]{Uhlmann-Stochasticity}%
  \BibitemOpen
  \bibfield  {author} {\bibinfo {author} {\bibfnamefont {P.}~\bibnamefont
  {Alberti}}\ and\ \bibinfo {author} {\bibfnamefont {A.}~\bibnamefont
  {Uhlmann}},\ }\href@noop {} {\emph {\bibinfo {title} {{S}tochasticity and
  {P}artial {O}rder: {D}oubly {S}tochastic {M}aps and {U}nitary {M}ixing}}}\
  (\bibinfo  {publisher} {Springer},\ \bibinfo {year} {1982})\BibitemShut
  {NoStop}%
\bibitem [{\citenamefont {Nakahara}\ \emph {et~al.}(2008)\citenamefont
  {Nakahara}, \citenamefont {Rahimi},\ and\ \citenamefont
  {Saitoh}}]{Nakahara-Decoherence}%
  \BibitemOpen
  \bibfield  {author} {\bibinfo {author} {\bibfnamefont {M.}~\bibnamefont
  {Nakahara}}, \bibinfo {author} {\bibfnamefont {R.}~\bibnamefont {Rahimi}}, \
  and\ \bibinfo {author} {\bibfnamefont {A.}~\bibnamefont {Saitoh}},\
  }\href@noop {} {\emph {\bibinfo {title} {{D}ecoherence {S}uppression in
  {Q}uantum {S}ystems}}}\ (\bibinfo  {publisher} {World Scientific},\ \bibinfo
  {year} {2008})\BibitemShut {NoStop}%
\bibitem [{\citenamefont {{P}etz}(2008)}]{Petz-QI}%
  \BibitemOpen
  \bibfield  {author} {\bibinfo {author} {\bibfnamefont {D.}~\bibnamefont
  {{P}etz}},\ }\href@noop {} {\emph {\bibinfo {title} {{Q}uantum {I}nformation
  {T}heory and {Q}uantum {S}tatistics}}}\ (\bibinfo  {publisher} {Springer},\
  \bibinfo {year} {2008})\BibitemShut {NoStop}%
\end{thebibliography}
\end{document}